\documentclass[journal]{ieeetran}
\IEEEoverridecommandlockouts

\usepackage{cite}
\usepackage{amsmath,amssymb,amsfonts}
\usepackage{algorithmic}
\usepackage{graphicx}
\usepackage{textcomp}
\usepackage{dcolumn}% Align table columns on decimal point
\usepackage{bm}% bold math
\usepackage{braket}
\usepackage{overpic}
\usepackage{subcaption}
\usepackage{mathtools}
\usepackage{hyperref}

\newtheorem{thm}{Theorem}
\newtheorem{rem}{Remark}

\def\BibTeX{{\rm B\kern-.05em{\sc i\kern-.025em b}\kern-.08em
    T\kern-.1667em\lower.7ex\hbox{E}\kern-.125emX}}
    
\begin{document}

\title{A Model Predictive Control-Inspired Quantum Algorithm}

\author{Dominic Messina, Alicia B. Magann, Mohan Sarovar, and Helen Durand \thanks{Dominic Messina and Helen Durand are with Wayne State University, Detroit, MI 48202 USA, and the Quantum Algorithms and Applications Collaboratory, Sandia National Laboratories, Albuquerque, NM 87123 USA.  Alicia Magann is with the Quantum Algorithms and Applications Collaboratory, Sandia National Laboratories, Albuquerque, NM 87123 USA.  Mohan Sarovar is with the Quantum Algorithms and Applications Collaboratory, Sandia National Laboratories, Livermore, CA 94550 USA.  Corresponding authors: Alicia Magann (email: abmagan@sandia.gov); Helen Durand (email: helen.durand@wayne.edu).}}

\maketitle

\begin{abstract}
We introduce a new hybrid quantum-classical algorithm inspired by an advanced control strategy known as model predictive control (MPC). This algorithm unifies the optimization-based design of variational quantum algorithms (VQAs) with the feedback-based design of feedback-based quantum algorithms (FQAs). Variational circuit parameters are optimized using a layer-wise receding horizon strategy, where observable measurements after every layer initialize a classically simulated dynamic model used to predict quantum state evolution and optimize over future parameterized gates. This hybrid algorithm can be used for applications such as ground state preparation and approximate combinatorial optimization, and presents an ideal use case for the integration of quantum computers with high-performance computing, where the latter resource can be used to increase the scale and efficiency of the predictions critical to MPC. We show through mathematical proof and numerical evidence that the MPC-based algorithm can be guaranteed to at least match the performance of FQAs. Through simulations on Max-Cut problems and a two-dimensional transverse-field Ising model, we demonstrate that relaxed implementations of the MPC-based algorithm can also provide improved performance in practice compared to an FQA. 
\end{abstract}

\begin{IEEEkeywords}
model predictive control, quantum computing, quantum algorithm, control theory
\end{IEEEkeywords}

\section{Introduction}\label{sec:Introduction}

Variational quantum algorithms (VQAs)~\cite{cerezo2021variational} such as the Quantum Approximate Optimization Algorithm (QAOA)~\cite{farhi2014quantumapproximateoptimizationalgorithm, hadfield2019quantum} and the Variational Quantum Eigensolver (VQE)~\cite{peruzzo2014variational} have been explored in the combinatorial optimization and computational chemistry domains, especially for noisy intermediate-scale quantum (NISQ)~\cite{preskill2018quantum} era quantum computers.  These algorithms utilize classical optimization to iteratively adjust the parameters of a quantum circuit to minimize a cost function. Feedback-based quantum algorithms (FQAs)~\cite{magann2022feedback, magann2022lyapunov} have been presented as an alternative strategy that employs measurement and feedback to explicitly determine circuit parameters. Both VQAs and FQAs, however, face limitations to their applicability. In VQAs, it is common for the optimization landscape to have flat and rugged regions that prevent the classical optimizer from finding optimal parameters in a reasonable timeframe. FQAs avoid the need for classical optimization, however they typically require longer circuits compared to VQAs. 

Managing optimality and classical optimization computation time is not unique to quantum computing.  In the field of chemical process control, for example, finite-horizon optimization-based control laws have been in use for decades to replace explicit yet potentially sub-optimal control laws, e.g., proportional-integral-derivative (PID) control laws~\cite{Ogunnaike}, with control actions that ideally approximate an infinite-horizon optimal policy~\cite{mayne2000constrained,grune2008infinite}. Though explicit optimal control laws, e.g., the linear quadratic regulator~\cite{kalman1960contributions}, were derived for special circumstances such as linear process models and a quadratic objective function, many optimal control problems have eluded explicit solutions~\cite{dontchev2020approximating} or been computationally challenging to solve for both high-dimensional continuous-time and discrete-time systems using the Hamilton-Jacobi-Bellman equations or dynamic programming. This motivated research evaluating how and when optimal control problems for higher-dimensional systems could be solved via approaches such as algorithmic and mathematical reformulations or approximations~\cite{chilan2020optimal,mceneaney2006curse,chow2018algorithm,horowitz2014linear,mceneaney2005curse,li2008mitigation,luus1990application,lawton1998numerically,powell2007approximate}.  A major approach to making online control tractable has been to move to receding-horizon control policies, which repeatedly solve for a sequence of control actions that are optimal over a short time period, re-solving over the same time period and adjusting future actions as the system evolves. Model predictive control (MPC)~\cite{ garcia1989model, mayne2000constrained, michalska2002robust, amrit2011economic, heidarinejad2012economic} has been an impactful receding horizon control strategy that seeks to obtain the benefits of online optimization and feedback control while accepting potential losses in optimality that may result from using a finite, receding horizon, compared to solving an infinite-horizon optimization problem.  The tradeoff between optimality and computational tractability seen in the comparison of VQAs and FQAs above demonstrates that near-term quantum algorithms are facing issues similar to those that led to widespread interest in MPC, suggesting that MPC may form a useful approach to overcoming limitations of VQAs and FQAs.

Various previous works have explored the applications of MPC in quantum control~\cite{hashimoto2016stability, goldschmidt2022model, clouatre2022model, lee2024robust, lee2024model, lee2025model, lee2025tractableinfinitehorizonstochasticmodel, guizani2025modelpredictivequantumcontrol}, but its potential applications in quantum algorithms have remained unexplored. Parameterized quantum circuits have strong parallels to control-theoretic concepts \cite{PRXQuantum.2.010101}, where parameterized gates can be interpreted as control actions that drive a quantum state to a target. In particular, the closed-form parameter-setting strategies used by FQAs are analogous to explicit feedback-based control, where control actions are conditioned on measurements of a system. One such example is the Feedback-based ALgorithm for Quantum OptimizatioN (FALQON)~\cite{magann2022feedback, magann2022lyapunov}, where parameterized quantum circuits are developed via a discrete-time analogy to quantum Lyapunov control. While feedback-based strategies bypass the classical optimization used in VQAs and can result in monotonic convergence to a (possibly local) optimum in the cost function, the gate parameterizations that they prescribe may not be globally optimal. MPC incorporates both feedback and optimization, suggesting that it could form a strategy that provides more optimal gate parameterizations than FQAs, which could in turn decrease the high circuit depth of the solutions output by FQAs.

In this work, we present a hybrid quantum-classical strategy for designing parameterized quantum circuits inspired by the layer-wise, feedback-based design of FALQON and the receding horizon, model-based optimization strategy of MPC. Specifically, we develop an optimization algorithm that uses classical simulation of quantum circuits to guide the selection of circuit parameters. The key innovation of this work is the demonstration of the use of advanced control methods in generating a hybrid quantum-classical algorithm that can reduce the gate depth needed to achieve target cost function values relative to FALQON. Our hybrid algorithm has the key feature of a tunable classical computing overhead since as we discuss below, longer prediction horizons can correspond to improved performance of an MPC algorithm. However, the classical simulation of a quantum circuit for prediction, even for a short time horizon, can be computationally challenging (even with the integration of approximate methods that we discuss below), and thus there is a natural trade-off between prediction and MPC algorithm quality and computational burden in our approach. A natural avenue for increasing the quality of prediction is the integration of high performance computing (HPC) with the quantum computer in order to perform the classical circuit simulation.

We first provide background on FALQON and MPC in Sec. \ref{sec:background}. We then outline the construction of a quantum algorithm for variational state preparation and energy estimation using MPC concepts in Sec. \ref{sec:MPCAlgorithm}. This section includes a theoretical proof that the performance of MPC-based quantum algorithms can be guaranteed to match or be superior to FALQON. In Sec. \ref{sec:PauliPropMPC} we discuss the formulation of reduced-order models to ease the computational burden of MPC and construct a concerete reduced-order model based on Pauli propagation and classical shadows. Then, Sec. \ref{sec:numerics} presents numerical simulations that demonstrate various aspects of the MPC-based quantum algorithm, including performance as a function of hyperparameters in the algorithm. These numerical simulations elucidate interactions between the various design choices that can be made in the pursuit of a tractable yet useful algorithm design.

\section{Background}
\label{sec:background}
The MPC-based algorithm takes inspiration from FALQON, where gate parameters are conditioned on feedback received prior to adding a layer. A major performance goal of the MPC-based algorithm is to reduce gate depth and improve convergence to the target solution relative to FALQON. As a result, in many of the simulations presented, the MPC-based algorithm is benchmarked against FALQON. In this section, we recall background on FALQON and MPC relevant to the construction of the MPC-based algorithm.

\subsection{Feedback-Based Quantum Algorithm}\label{sec:FALQON}

FALQON~\cite{magann2022feedback,magann2022lyapunov} is a feedback-based quantum algorithm for approximate ground state preparation. FALQON uses concepts from quantum Lyapunov control~\cite{cong2013survey} to determine a sequence of gate parameters, deriving a circuit that produces a state that (ideally) minimizes the expectation value of a problem Hamiltonian $H_p$, where the minimum value is the ground state energy. The gate parameter update strategy is derived from quantum Lyapunov control arguments applied to a quantum system with dynamics described by the Schr\"{o}dinger equation as follows:
\begin{equation}
\begin{aligned}
    i\hbar\frac{d}{dt}\ket{\psi(t)} & = H(t)\ket{\psi(t)} \\
                                & = [H_p + H_d\beta(t)]\ket{\psi(t)}, \label{eq:Schrodinger}
\end{aligned}
\end{equation}
where $H_p$ represents the drift Hamiltonian, $H_d$ is the driver Hamiltonian that is multiplied by the time-dependent control function $\beta(t)$, $\ket{\psi(t)}$ represents the quantum state at time $t$, and $\hbar = 1$ is the reduced Planck constant. The control function $\beta(t)$ can be chosen according to quantum Lyapunov control principles \cite{grivopoulos2003lyapunov} to monotonically decrease the expectation value of $H_p$
by selecting $\beta(t) = -A(t)$, where $A(t) \equiv \braket{\psi(t)|i[H_d,H_p]|\psi(t)}$, such that:
\begin{equation}
    \frac{d}{dt}\langle\psi(t)|H_p|\psi(t)\rangle = -A(t)^2 \leq 0.
\end{equation}
The strategy by which FALQON develops a parameterized circuit for ground state approximation is inspired by quantum simulation of Eq. \eqref{eq:Schrodinger} under the control policy above using Trotterization~\cite{trotter1959product}. Each ``layer'' of FALQON involves time evolution operators of the form $U_p = e^{-iH_p\Delta_t}$ and $U_d = e^{-iH_d\beta_k\Delta_t}$, where $k = 1, 2, ..., f$ denotes the FALQON layer and $\Delta_t$ represents the step size, resulting in a circuit of the form $U_d(\beta_f)U_p...U_d(\beta_1)U_p$. 

FALQON can be implemented through the following strategy:
\begin{enumerate}
	\item \label{step:FALQONAlgoStep1} Initialize qubits to a state $\ket{\psi_0}$. Set $\beta_1 = -A_0$, where the subscript denotes the circuit layer, and append $U_p$ and $U_d(\beta_1)$ to the quantum circuit.
	\item \label{step:FALQONAlgoStep2} For each subsequent step $k$, run the current quantum circuit and perform measurements to estimate $A_{k-1} = \langle \psi_{k-1}|i[H_d,H_p]|\psi_{k-1}\rangle$. Set $\beta_{k} = -A_{k-1}$ and append layer $k$ (consisting of $U_p$ and $U_d(\beta_{k})$) to the circuit.
	\item \label{step:FALQONAlgoStep3} Terminate after a fixed number of layers $f$.
\end{enumerate}
The premise of adaptively growing a quantum circuit is central to FALQON, and has previously been utilized in other algorithms, notably ADAPT-VQE and its variations \cite{grimsley2019adaptive,PRXQuantum.2.020310,PhysRevResearch.4.033029,PhysRevResearch.6.013254,x8g1-7h1k,sambasivam2025tepidadaptadaptivevariationalmethod}, algorithms for quantum imaginary time evolution \cite{motta2020determining,rw81-k8vk}, algorithms based on Riemannian gradient descent \cite{PhysRevA.107.062421,PhysRevResearch.5.033227,ht2m-1j91}, and algorithms utilizing layer-wise variational parameter optimization \cite{skolik2021layerwise,PhysRevA.104.L030401}. A distinguishing feature of FALQON against this backdrop is its explicit connection to quantum Lyapunov control theory, and compared to many other algorithms, its lack of any classical optimization requirement. Since the initial development of FALQON, a variety of new ideas, modifications, and extensions have enabled new theoretical and numerical analyses, performance improvements, reduced costs, and broader applicability \cite{PhysRevResearch.6.033336, clausen2023measurement, PhysRevResearch.6.043068, abdul2024adaptive, brady2024focqsfeedbackoptimallycontrolled, snht-7jsf, PhysRevResearch.7.013035, qc91-5mj2, vanlong2025imaginarytimeenhancedfeedbackbasedquantumalgorithms, chen2025lyapunovframeworkquantumalgorithm,  mozakka2026acceleratingfeedbackbasedalgorithmsquantum, perez2026learningparametercurvesfeedbackbased, rahman2026feedback, rattighieri2026measurementguidedstaterefinementshallow, 10.1088/2058-9565/ae7d4e, mancini2026optimalfalqonquantumapproximate}.

\begin{rem}
The same circuit structure that alternates between applications of $U_p$ and $U_d$ used in FALQON is also used in QAOA. However, in QAOA, the $U_p$ operations are dependent on an additional set of parameters, $\gamma_1,\cdots,\gamma_f$, resulting in a circuit of the form $U_d(\beta_f)U_p(\gamma_f)\cdots U_d(\beta_1)U_p(\gamma_1)$. Additionally, QAOA is an optimization-based hybrid quantum-classical algorithm, where $\braket{H_p}$ is minimized by solving for optimal values of all the circuit parameters, $\gamma_1,\cdots,\gamma_f$ and $\beta_1,\cdots,\beta_f$, via classical optimization.
\end{rem}

\subsection{Model Predictive Control}\label{sec:MPC}

Model predictive control (MPC)~\cite{garcia1989model,mayne2000constrained} is an optimization-based control strategy that operates in a {closed loop}, where control parameters are determined via on-line feedback control. The control policy determined by an MPC is the optimal value of the decision variables obtained from an optimization problem, where the objective function and constraints depend on predictions of states determined by a dynamic model over a \textit{prediction horizon} of $N$ time steps into the future. Though in principle, any reasonable model of a dynamic system could be used in an MPC, many research works utilizing MPC come from engineering fields (e.g., \cite{lages2006real,kang2009linear,biegler2015advances}) in which the system dynamic model may be written as a system of difference equations in discrete time:
\begin{equation} \label{eq:ClassOfSystems}
\begin{aligned}
\psi_{k} = f(\psi_{k-1},\beta_k)
\end{aligned}
\end{equation}
where $\psi_k$ denotes the state of the system at time step $k$ and $\beta_k$ denotes the input applied to the system at time step $k$. MPC is a \textit{receding horizon} control strategy, where the optimization problem is re-solved at every step $k$, $k=1,\ldots$ as the system operates, each time predicting over $N$ future time steps, $k,k+1,\cdots,k+N-1$. Feedback is incorporated at every time step $k$ by using a measurement of the state $\psi_{k-1}$ to initialize a dynamic model that is used to predict the system response under a sequence of inputs, denoted as follows:
\begin{equation}
    \vec{\beta}_{(i|k)} = \big[\beta(k|k), \beta(k+1|k), \cdots \beta(k+N-1|k) \big] \label{eq:betavec}
\end{equation}
which represents inputs evaluated at time step $k$ that would be implemented in the future time steps $i=k,\ldots,k+N-1$. Optimal values of the $N$ inputs, $\vec{\beta}^*_{(i|k)}$, are selected to minimize a cost function dependent on predictions of the state under the inputs. The first input in the sequence:
\begin{equation}
    \beta^*_k \equiv \beta^*(k|k) \label{eq:betastar}
\end{equation}
that evolves the state $\psi_{k-1}$ to $\psi_k$ is applied to the system. At the next time step $k+1$, new measurements are taken and the optimization problem is re-solved, again looking $N$ time steps into the future, but shifting the horizon forward to extend from time step $k+1$ to time step $k+N$. The time difference between consecutive time steps is the {step size} $\Delta_t$.

The optimization problem of an MPC for the system of Eq. \eqref{eq:ClassOfSystems} is formulated as follows:
\begin{subequations} \label{eq:MPCEqn}
	\begin{align}
	\min_{\vec{\beta}_{(i|k)}}\hspace{3mm} & \sum_{i=k}^{k+N-1} L(\tilde{\psi}_{i},\beta_{i}) \,   \label{eq:MPCEqn:Objective} \\
	\text{s.t.} \hspace{3mm} & \tilde{\psi}_{i} = f(\tilde{\psi}_{i-1},\beta_i), \, \nonumber \\
    & \quad i = \{k, k+1,...,k+N-1\}  \label{eq:MPCEqn:Model} \\
	& \tilde{\psi}_{k-1} = \psi_{k-1}  \label{eq:MPCEqn:Measurement}
	\end{align}
\end{subequations}
where the objective function to be minimized consists of the sum of values of a cost function $L(\cdot,\cdot)$, evaluated at each time step in the prediction horizon. The state predictions $\tilde{\psi}_i,~ i = \{k, k+1,...,k+N-1\}$ are obtained from the model of Eq. \eqref{eq:MPCEqn:Model}, initialized by a state measurement indicated by Eq. \eqref{eq:MPCEqn:Measurement}. An illustrative example of MPC is shown in Figure~\ref{fig:MPCFigure}.

\begin{figure}
    \centering
    \includegraphics[width=\columnwidth]{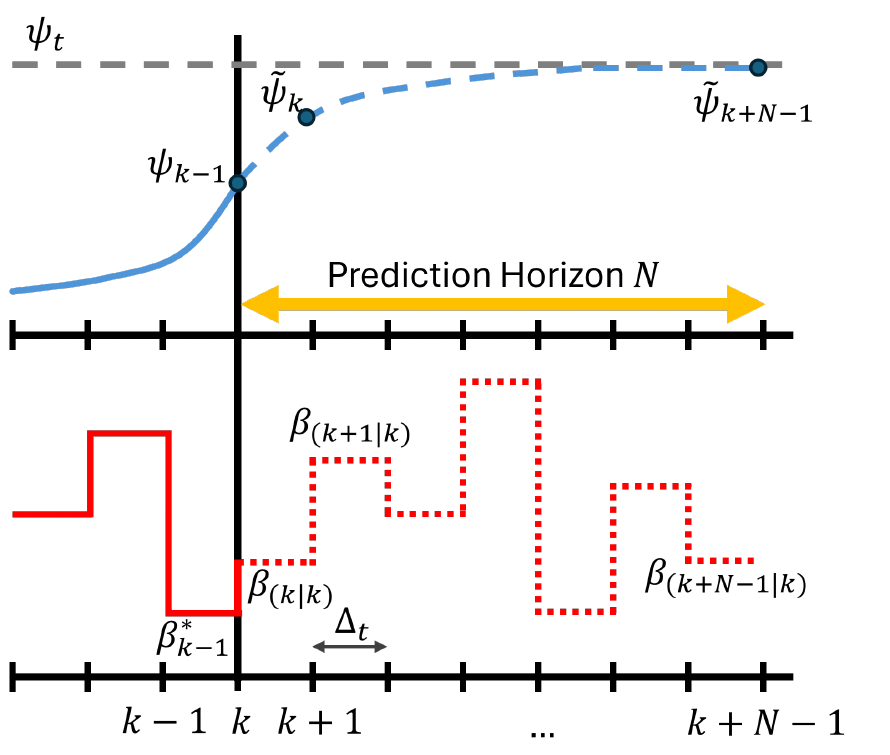}
    \caption{\textbf{Model predictive control illustration.} The state trajectory is shown above in blue and the input trajectory is shown below in red. The solid blue line indicates the state trajectory under the inputs applied to the system, shown as solid red lines, while dashed lines indicate the predicted state and input trajectories. At time step $k$, a measurement of $\psi_{k-1}$ initializes a dynamic model that predicts the evolution of the state over the prediction horizon consisting of $N$ time steps from $k$ to $k+N-1$, guiding the optimization process in solving for an input vector $\vec{\beta}_{(i|k)}$ that minimizes an objective function. Following a receding horizon strategy, the first optimal input $\beta^*_k$ is applied to the process, and the optimization problem is re-solved at the next time step $k+1$ with a measurement of $\psi_k$. This illustration shows an example of an MPC that is driving a system towards a target state $\psi_{t}$.}
    \label{fig:MPCFigure}
\end{figure}

There are several features of MPC that make it attractive in control engineering and motivate our interest in exploiting its properties to design a new control-inspired quantum algorithm: 1) Though MPC uses optimization to determine control inputs, it also is a feedback control law.  This means that the model used in the MPC does not need to be perfect for it to produce reasonable control actions in many cases, but only needs to be sufficiently accurate to enable the controller to find control actions that would minimize the objective function and meet the constraints. 2) The prediction horizon is a tuning parameter of the MPC strategy that permits the MPC to form a desired finite-time approximation to the infinite-horizon optimal control problem. In practice, it enables a tradeoff between computation time and suboptimality with respect to the infinite-horizon optimal control problem.  Longer prediction horizons might provide improved closed-loop performance at the expense of increased computation time, whereas shorter prediction horizons may result in more myopic decision-making but reduce computation time. 3) MPC can be designed to achieve a number of different behaviors due to the freedom in selecting constraints and objective functions.  For example, the objective function can be selected to encourage the system state to move toward a desired setpoint, track a predetermined reference trajectory, or optimize economic performance~\cite{amrit2011economic}. Constraints can be added to the MPC optimization problem to meet objectives or to enable conditions to be derived under which theoretical properties, such as stability of a steady-state of the closed-loop system of Eq. \eqref{eq:ClassOfSystems}, can be proven. An example of this that we extend to the MPC-based algorithm in Sec. \ref{sec:Performance} is the addition of a terminal constraint that requires the predicted state to converge to a specified state at the end of the prediction horizon~\cite{ellis2014tutorial}.

\section{Model Predictive Control-Based Quantum Algorithm}\label{sec:MPCAlgorithm}

\begin{figure*}[t]
    \centering
    \setlength\fboxsep{0pt} % No padding around the box
    \setlength\fboxrule{0.25pt} % Thickness of the box border
    \includegraphics[width=6.8in]{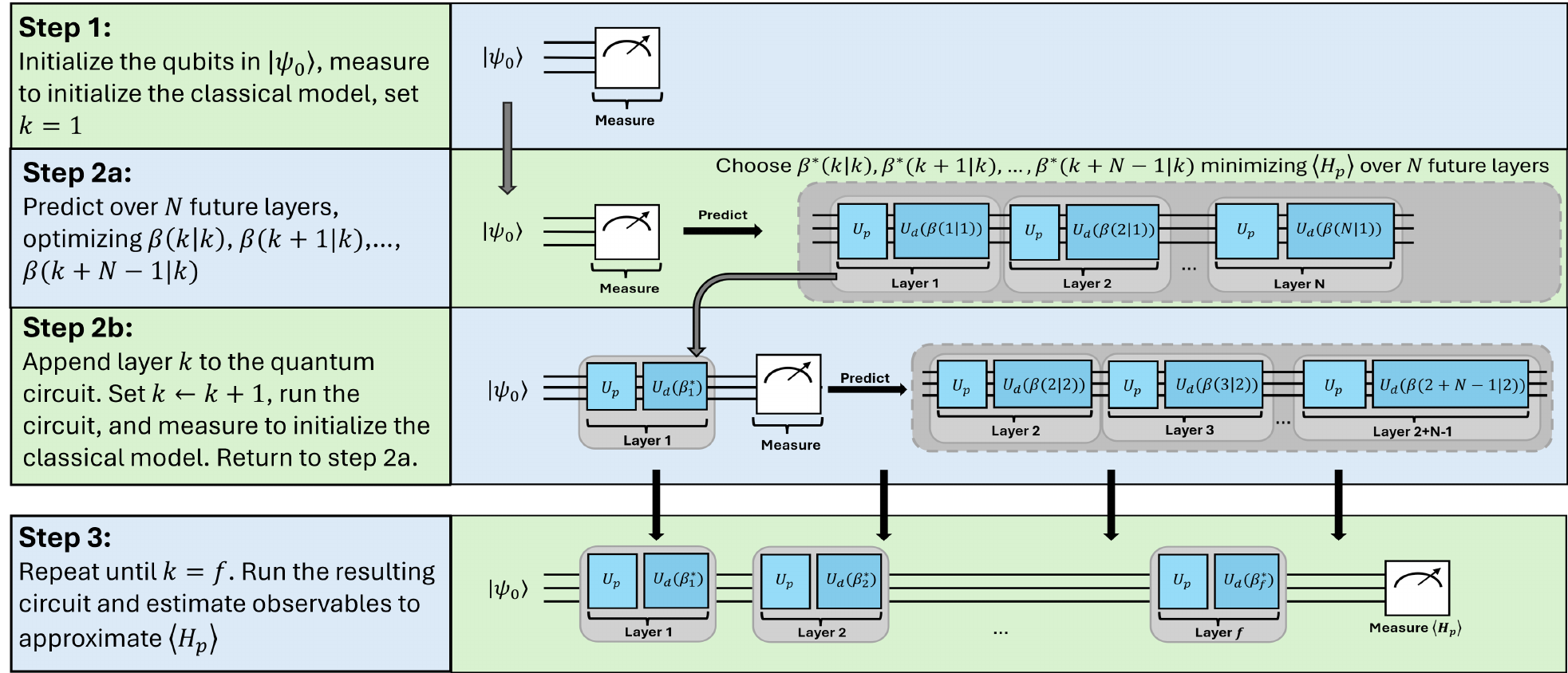}
    \caption{\textbf{MPC-based quantum algorithm procedure.} The initial step proceeds by initializing the quantum state as $\ket{\psi_0}$ and using that to initialize the classical model in the MPC used to obtain $\beta_1^*$. Layer $k = 1$ then is implemented to evolve the quantum state to $\ket{\psi_1} = U_d(\beta_1^*)U_p\ket{\psi_0}$. The quantum state $\ket{\psi_1}$ is then repeatedly measured to provide a reconstruction used to initialize the dynamic model that predicts the evolution of the quantum state over $N$ future layers. A classical optimizer solves for optimal values of $\vec{\beta}_{(i|2)}$, i.e., $\beta^*(2|2), \beta^*(3|2),\cdots,\beta^*(2+N-1|2)$, and following the receding horizon strategy, layer $k=2$ parameterized by $\beta_2^*$ is added to the circuit. This procedure is followed for every subsequent layer $k =3,4,\cdots,f$, where $\ket{\psi_0}$ is prepared and evolved by the circuit through layer $k-1$ to produce the quantum state $\ket{\psi_{k-1}} = U_d(\beta_{k-1}^*)U_p\cdots U_d(\beta_1^*)U_p\ket{\psi_0}$ which is repeatedly measured to provide a reconstruction of $\ket{\psi_{k-1}}$ that initializes the dynamic model. The model then predicts the evolution of the quantum state starting from $\ket{\psi_{k-1}}$ through layers $k,k+1,\cdots,k+N-1$, under the gate parameters $\vec{\beta}_{(i|k)}$.  A classical optimization solver is used to select optimal values of the gate parameters $\vec{\beta}^*_{(i|k)}$, and layer $k$ with parameter $\beta^*_k$ is added to the circuit. After a prescribed number of layers $f$ or after a convergence criterion is met, the quantum circuit can be sampled to approximate $\braket{H_p}_f$.}
    \label{fig:MPCQuantumAlgorithm}
\end{figure*}

Model predictive control has unique features as a control strategy, blending feedback control and optimal control concepts. The combined feedback and finite-horizon optimization capabilities of MPC suggest that it may form an alternative to FQAs in cases where optimized gate parameterizations are more desirable than the FQA policy, and VQAs in cases where optimizing over only part of the circuit reduces stalling by the classical optimizer. This positions an MPC-based algorithm as a new tool that combines the respective strengths of VQAs and FQAs, optimization-based parameter selection and integrated feedback, and has potential to handle their weaknesses such as classical optimization challenges for VQAs and excessive circuit depth for FQAs.

While the MPC-based quantum algorithm developed here can be used with any parameterized quantum circuit ansatz, in the following we work with the same ansatz used by FALQON in order to compare performance against this FQA. Our MPC-based algorithm iteratively constructs a parameterized quantum circuit in a layer-wise fashion, with each layer consisting of alternating applications of $U_p$ and $U_d$. In each iteration $k$, an optimization problem is solved where predictions of the evolution of the quantum state $\ket{\psi_{k-1}}$ over $N$ future layers with a set of parameters $\vec{\beta}_{(i|k)}$ is used to evaluate an objective function. When optimal values of $\vec{\beta}_{(i|k)}$ are found, layer $k$ with the optimal parameter $\beta_k^*$ is added to the circuit, where $\vec{\beta}_{(i|k)}$ and $\beta_k^*$ mirror the notation of Eqs. \eqref{eq:betavec} and \eqref{eq:betastar}, respectively. This results in circuits of the form $U_d(\beta^*_f)U_p...U_d(\beta^*_1)U_p$ for $f$ total layers. The base-case formulation is a straightforward extension of MPC outlined in Eq. \eqref{eq:MPCEqn} to a quantum algorithm, where the optimization problem is defined as:
\begin{subequations} \label{eq:MPCAlgorithm}
	\begin{align}
	\min_{\vec{\beta}_{(i|k)}} ~ & J_k\label{eq:MPCAlgorithm:Objective} \\
	\text{s.t.} \hspace{3mm} & \ket{\tilde{\psi}_{i}} = U_d(\beta_{i})U_p\ket{\tilde{\psi}_{i-1}}, \nonumber \\
    & \quad \quad ~i\in\{k,\ldots,k+N-1\} \label{eq:MPCAlgorithm:Model} \\
	& \ket{\tilde{\psi}_{k-1}} = \ket{\psi_{k-1}}  \label{eq:MPCAlgorithm:Measurement}
	\end{align}
\end{subequations}
where $\ket{\tilde{\psi}_i}, i \in \{k, k+1,\cdots,k+N-1\}$ denotes the predictions of the quantum state at layer $i$ made by the model in the time step where layer $k$ is added to the circuit. Eq. \eqref{eq:MPCAlgorithm:Measurement} indicates that the model of Eq. \eqref{eq:MPCAlgorithm:Model} is initialized by $\ket{\psi_{k-1}}$, where $\ket{\psi_{k-1}} = U_d(\beta^*_{k-1})U_p\cdots U_d(\beta^*_{1})U_p \ket{\psi_0}$. In Eq. \eqref{eq:MPCAlgorithm:Objective}, $J_k$ represents the objective function for the problem, which may be a function of the state predictions from Eq. \eqref{eq:MPCAlgorithm:Model} and gate parameters $\vec{\beta}_{(i|k)}$. We denote the expectation value of $H_p$ at layer $k$ as $\braket{H_p}_k = \bra{\psi_k}H_p\ket{\psi_k}$ and the predicted value of $\braket{H_p}_k$ as $\tilde{\braket{H_p}}_k = \bra{\tilde{\psi}_k}H_p\ket{\tilde{\psi}_k}$. 
The algorithm is depicted in Figure~\ref{fig:MPCQuantumAlgorithm} with a description in the caption.

As stated, the optimization problem performed at every layer, per Eq. \eqref{eq:MPCAlgorithm:Objective}, rapidly becomes intractable as the dimension of the quantum system grows. It requires full characterization of the output state at layer $k-1$ in order to set the model initial state in Eq. \eqref{eq:MPCAlgorithm:Measurement}. It also requires full state vector simulation to propagate the dynamics, per Eq. \eqref{eq:MPCAlgorithm:Model}. In Sec. \ref{sec:PauliPropMPC} we describe how this optimization can be performed in a tractable manner through the use of reduced-order models. For the remainder of this section we keep this original formulation in order to demonstrate some key properties of the MPC-based algorithm.

We note that there are two key hyperparameters in the algorithm, as defined above and in Fig. \ref{fig:MPCQuantumAlgorithm}. The first is $\Delta_t$, which defines the ``time step'' since $U_d$ and $U_p$ are evolutions proportional to this parameter, and the second is $N$, which defines the prediction time horizon. The fact that MPC performs a non-local optimization over a finite prediction horizon means that unlike for FALQON and other FQAs, the choice of $\Delta_t$ cannot be informed by considerations of preserving local continuity and minimizing Trotter error \cite{magann2022feedback, magann2022lyapunov}. Instead, the choice of these parameters is heuristic and in the numerical demonstrations in Sec. \ref{sec:Hyperparameters}, we explore approaches for choosing these hyperparameters.

\subsection{Choice of objective function}
\label{sec:objfxn}
The choice of the objective function $J_k$ in Eq. \eqref{eq:MPCAlgorithm} is an important consideration. For the circuit ansatz used in the specification of the MPC-based quantum algorithm above, the goal is to minimize the expectation value of the problem Hamiltonian, $\braket{H_p}$. To achieve this, the MPC could consider a variety of objective functions. One is a \emph{terminal cost} function, resulting in minimization of $\langle H_p \rangle$ at the end of the prediction horizon, which takes the following form:
\begin{equation}
    J_{T,k} = \tilde{\braket{H_p}}_{k+N-1}. \label{eq:TerminalCost}
\end{equation} 
We refer to the MPC-based algorithm under this objective function as the \emph{terminal cost MPC}. Another idea is to select an objective function more closely resembling the objective function used in traditional applications of MPC, which seeks to minimize the sum of the deviations of the system state from a target value over the prediction horizon. In the MPC-based algorithm, a target quantum state is not known, but penalizing a sum of expectation values of $H_p$ at the states produced at the end of each layer in the prediction horizon achieves a similar effect. The resulting {cumulative cost} has the following form:
\begin{equation}
    J_{C,k} = \sum_{i=k}^{k+N-1}\tilde{\braket{H_p}}_{i}. \label{eq:CumulativeStageCost}
\end{equation}
We refer to the MPC-based algorithm under this objective function as the \emph{cumulative cost MPC}.

This freedom in defining the MPC objective function can be exploited to tune the performance of the algorithm. In Sec. \ref{sec:numerics}, we numerically compare the two objective functions defined above to demonstrate how this choice affects performance, with subsequent analyses focusing primarily on the terminal cost objective function, Eq. \eqref{eq:TerminalCost}. We leave it as an avenue for future work to systematically explore the range of possible objective functions.

\subsection{A lower bound on performance}
\label{sec:Performance}
Before numerically demonstrating the MPC-based quantum algorithm defined above, we first rigorously derive a lower bound on its performance. We specifically prove that the MPC-based algorithm can be guaranteed to at least match FALQON in terms of objective function value at a fixed number of layers. This can be achieved through modifying the MPC optimization problem by adding a constraint based on the state that would be achieved by FALQON at the end of the prediction horizon, which we term the \textit{terminal constraint}.

We first define notation for the quantum state predictions under FALQON and under the MPC-based algorithm with the terminal constraint.  Denoting the quantum state at layer $k$ as $\ket{\phi_k}$ and gate parameter under FALQON as $\nu_{k}$, we use the following notation to signify the evolution of the quantum state under FALQON initialized at $\ket{\phi_0}$:
\begin{equation}
    \begin{aligned}
        \ket{\phi_{k}} & = U_d(\nu_{k})U_p\ket{\phi_{k-1}} \\
        \nu_{k} & = -\braket{\phi_{k-1}|i[H_d,H_p]|\phi_{k-1}}\label{eq:FALQONDiscrete}
    \end{aligned}
\end{equation}
We define the state of the quantum system operated under the terminal constraint MPC-based algorithm as $\ket{\psi_{k}}$, initialized at $\ket{\psi_{0}} = \ket{\phi_0}$. Therefore, the state at the end of layer $k$ under the MPC is given by:
\begin{equation}
\begin{aligned}
    \ket{\psi_{k}} & = U_d(\beta^*_{k})U_p\ket{\psi_{k-1}} \\
    \beta^*_{k} & = \beta^*(k|k) \label{eq:NominalDiscrete}
\end{aligned}
\end{equation}
where both algorithms use the same step size $\Delta_t$.

Using this notation, the terminal constraint MPC-based algorithm optimization problem is defined as follows:
\begin{subequations} \label{eq:LMPCAlgorithm}
	\begin{align}
	\min_{\vec{\beta}_{(i|k)}} ~ & J_k \label{eq:LMPCAlgorithm:Objective} \\
	\text{s.t.} \hspace{3mm} & \ket{\tilde{\psi}_{k}} = U_d(\beta_{k})U_p\ket{\tilde{\psi}_{k-1}} \label{eq:LMPCAlgorithm:Model} \\
	& \ket{\tilde{\psi}_{k-1}} = \ket{\psi_{k-1}}  \label{eq:LMPCAlgorithm:Measurement} \\
    & \ket{\tilde{\psi}_{k+N-1}} = \ket{\phi_{k+N-1}}
    \label{eq:LMPCAlgorithm:TerminalConstraint} 
	\end{align}
\end{subequations}
where $\ket{\tilde{\psi}_i}, i \in \{k, k+1,\cdots,k+N\}$ denotes the predictions of the quantum state at layer $i$ operated under the terminal constraint MPC-based algorithm.  Compared to Eq. \eqref{eq:MPCAlgorithm} the additional constraint of Eq. \eqref{eq:LMPCAlgorithm:TerminalConstraint} requires that the predicted quantum state at the end of the prediction horizon under the MPC-based algorithm, $\ket{\tilde{\psi}_{k+N-1}}$, must be equal to the state under FALQON after layer $k+N-1$, $\ket{\phi_{k+N-1}}$.

A key aspect of the terminal constraint MPC strategy is the use of a shrinking prediction horizon, where the horizon decreases in the last $N$ sampling periods before time step $f$, i.e., $N_k = N - j + 1$ and $j = k - f + N$ for $k \in \{f-N+1,\dots,f\}$, such that the MPC does not predict beyond layer $f$. Regardless of the choice of the objective function of Eq.~\eqref{eq:LMPCAlgorithm:Objective}, the terminal constraint then requires that $\ket{\psi_f} = \ket{\phi_f}$, so that the value of $\langle H_p \rangle_f$ under the MPC will be the same as that under FALQON.  If the conditions under which FALQON results in a monotonic decrease of $\braket{H_p}$ are met~\cite{magann2022lyapunov}, $\ket{\phi_f}$ is the state with the most optimal value of $\braket{H_p}$ prepared by FALQON. In this case, it is guaranteed that the terminal constraint MPC will find the most optimal value of $\braket{H_p}$ that is found under FALQON in the same number of layers. A feature of the cumulative cost objective function of Eq.~\eqref{eq:CumulativeStageCost} is that it incentivizes a decrease in $\braket{H_p}$ over the prediction horizon, potentially motivating an improvement of $\braket{H_p}$ in the MPC prior to layer $f$ compared to the final value obtained by FALQON.

Motivated by this, under the terminal constraint MPC of Eq. \eqref{eq:LMPCAlgorithm} and cumulative cost objective function of Eq.~\eqref{eq:CumulativeStageCost}, the following theorem, adapted from~\cite{durand2016economic}, guarantees that the cumulative sum of $\braket{H_p}$ over $f$ total layers is upper bounded by the same measure under FALQON.

\begin{thm}\label{thm:terminalconstraintLMPC}
Consider the system of Eq. \eqref{eq:NominalDiscrete} under the terminal constraint MPC-based algorithm of Eq. \eqref{eq:LMPCAlgorithm} based on FALQON with the cumulative cost objective function of Eq. \eqref{eq:CumulativeStageCost}. Let $\Delta_t > 0$. For any strictly positive finite integer $f$, when a shrinking horizon is employed over the last $N$ sampling periods before time step $f$, the closed-loop performance of the terminal constraint MPC-based algorithm is bounded as follows:
\begin{equation}
\sum_{k=1}^{f} \langle \psi_{k} | H_p|\psi_{k} \rangle_{\beta_k^*} \leq \sum_{k=1}^{f} \langle \phi_{k} | H_p|\phi_{k} \rangle_{\nu_k} \label{eq:prop1Option2}
\end{equation}
\end{thm}

A detailed proof is presented in Appendix~\ref{app:LMPCProof}, but here, we indicate why the bound in Eq. \eqref{eq:prop1Option2} is reasonable conceptually. The terminal constraint requires the predicted quantum state to be the state under FALQON at the end of the prediction horizon. 
At layer 1, the parameterization under FALQON is feasible and satisfies the terminal
constraint. 
If this baseline solution is selected by the MPC, then $\beta^*(p|1) = \nu_p~\forall ~p \in  \{ 1,\dots, N \}$ and $\ket{\psi_p} = \ket{\phi_p}~\forall ~p \in \{ 1,\dots, N \}$ such that $\sum_{i=1}^{N} \braket{\tilde{\psi}_{i}|H_p|\tilde{\psi}_{i}}_{\beta^*(i|1)} = \sum_{i=1}^{N} \braket{\phi_{i}|H_p|\phi_{i}}_{\nu_i}$, providing an upper bound on the sum of $\braket{H_p}_p$ from layer $p=1$ to layer $p=N$. At layer 1, the terminal constraint MPC seeks to find a gate parameterization that results in a more optimal value of the cumulative cost than the parameterization under FALQON for layers $1$ through $N$ while satisfying the terminal constraint. If a more optimal set of parameters is found at layer 1 where $\beta^*(p|1) \neq \nu_p$, then $\ket{\psi_p} \neq \ket{\phi_p}$, $p \in \{ 1,\ldots,N \}$, and the parameterization under FALQON is no longer a valid solution to the optimization problem at layer $2$. However, the optimal solution from layer $1$, $\vec{\beta}_{(i|1)} = \vec{\beta}^*_{(i|1)}, ~i \in \{2, \dots,N-1\}$ is a valid solution for layers $2$ through $N-1$, and since the terminal constraint required $\ket{\tilde{\psi}_{N}} = \ket{\phi_{N}}$ at layer $1$, the parameter under FALQON at layer $N+1$, i.e., $\beta(N+1|2) = \nu_{N+1}$, is a solution for the last layer in the prediction horizon that satisfies the terminal constraint. This parameterization then provides a baseline solution to the optimization problem at layer $2$, where again a more optimal set of parameters is sought.  Applying this recursively, the MPC-based algorithm seeks to find a qubit state trajectory that further minimizes Eq. \eqref{eq:LMPCAlgorithm:Objective} compared to the trajectory under FALQON from layer $k$ to $k+N-1$ while satisfying the terminal constraint. As a result, it is possible that states with lower values of $\braket{H_p}$ may be prepared, and these states may be prepared with reduced numbers of layers relative to FALQON. This can be encouraged by prescribing a large prediction horizon, where the MPC-based algorithm can utilize more layers in the beginning of the horizon decreasing $\tilde{\braket{H_p}}$ before using the remaining layers to evolve the predicted state such that the constraint is met.

\section{Reduced-order models for system dynamics}
\label{sec:PauliPropMPC}
The MPC-based quantum algorithm formulated in Eq. \eqref{eq:MPCAlgorithm} requires predicting the output state of each layer in the prediction horizon under the dynamic model as the circuit parameters are optimized. Completely characterizing the output state and accurately propagating the dynamic model rapidly become intractable with growth in the number of qubits in the model, $n$. A strategy for carrying out this classical simulation is to develop reduced-order models for the system. There are several possibilities for doing this, and in this section we describe a promising approach based on Pauli propagation \cite{rudolph2025pauli} (or Sparse Pauli Dynamics (SPD)~\cite{beguvsic2025real}) and classical shadows \cite{huang2020predicting}. 

The key observation is that the MPC optimization ultimately only requires the predicted value of the problem Hamiltonian at time points in the prediction horizon, $\braket{\tilde{H}_p}_i$ for $i\in\{k, ... ,k+N-1\}$. For terminal cost MPC, only the value $\braket{\tilde{H}_p}_{k+N-1}$ is required. The first step in exploiting this observation is to decompose $H_p$ in the Pauli basis, as
\begin{equation}
    H_p = \sum_{P \in \mathcal{P}} a_{P} P,
    \label{eq:HpPaulis}
\end{equation}
where $\mathcal{P}$ is the set of $n$-qubit Pauli operators and $a_{P} = {\textrm{tr}}(H_p P)$. Then, $\braket{H_p}_{k+N-1} = \sum_{P \in \mathcal{P}} a_{P} \braket{P}_{k+N-1}$, and as a result, predicting the expectation value of $H_p$ can be transformed into the task of predicting the expectation value of the Pauli operators $P$. Notably, for common problem Hamiltonians that derive from optimization problems or spin models, the number of Pauli operators in this expansion (or equivalently, the number of weights $a_P\neq 0$) is $O(\textrm{poly}(n))$. 

To predict the expectation values of the Pauli operators, we can propagate them \textit{backwards} through the circuit. To see this, let us explicitly write out the expectation value of a Pauli operator: 
\begin{equation}
\begin{aligned}
    {\braket{P}}_{k+N-1} &=
    \bra{\psi_{k-1}}U_p^\dagger U_d^\dagger(\beta_{k})\cdots 
    U_p^\dagger U_d^\dagger(\beta_{k+N-1})~P \\
    &\quad\quad \times~U_d(\beta_{k+N-1})U_p \cdots U_d(\beta_{k})U_p\ket{\psi_{k-1}}. \label{eq:PauliPropEq1}
\end{aligned}
\end{equation}
The unitaries conjugating the Pauli can be written as products of parametrized unitaries generated by $n$-qubit Pauli operators, i.e., $U_\sigma(\theta) = e^{-i\theta \sigma}$, for $\sigma \in \mathcal{P}$. For example, for $H_p=c_1Z_1Z_2 + c_2Z_3Z_4$, we obtain $U_p=e^{-i\Delta_t c_1 Z_1Z_2}e^{-i\Delta_t c_2 Z_3Z_4}$. Conjugation of the Pauli operator $P$ by such unitaries can be computed by using properties of the Pauli group, 
\begin{equation}
    \begin{aligned}
        U_\sigma^\dagger(\theta)P U_\sigma(\theta) = \begin{cases}
            P & \text{if}~[\sigma, P] = 0 \\
            \cos(2\theta)P + i\sin(2\theta)\sigma P & \text{otherwise} \label{eq:PauliPropagation}
        \end{cases}
    \end{aligned}
\end{equation}
Note that $\sigma P \in \mathcal{P}$, and so this conjugation results in a linear combination of Pauli operators. 

Repeating this calculation allows one to express the expectation value of the Pauli operator under the state at layer $k+N-1$ in terms of a linear combination of Pauli operator expectation values under the state at layer $k-1$:
\begin{equation}
\begin{aligned}
    {\braket{P}}_{k+N-1} &= \sum_{P \in \mathcal{P}} c_P(\Delta_t, \beta_k,...\beta_{k+N-1}) \braket{P}_{k-1},
    \label{eq:Pauliexp}
\end{aligned}
\end{equation}
with $c_P(\Delta_t, \beta_k,...\beta_{k+N-1}) \in \mathbb{C}$. This expansion has $\textrm{poly}(n)\times 2^{O(N)}$ terms, and it can be computed in $O(e^N)$ time. This technique trades the exponential cost in $n$ that a state vector simulation of the model would incur for an exponential cost in $N$, the prediction horizon. Once this expansion is computed for every Pauli operator in the expansion of $\braket{H_p}_{k+N-1}$, per Eq. \eqref{eq:HpPaulis}, $\braket{H_p}_{k+N-1}$ can be computed in terms of the Pauli expectation values $\braket{P}_{k-1}$. 
The latter quantities can be measured directly under the state $\ket{\psi_{k-1}}$. For small prediction horizons, these expectations can be efficiently measured using classical shadows \cite{huang2020predicting} or its deterministic version \cite{PhysRevLett.127.030503} if the weight of the Paulis entering the expansion Eq. \eqref{eq:Pauliexp} is small. An explicit calculation of the Pauli propagation-based calculation of $\braket{H_p}$ for a Max-Cut example is given in Appendix~\ref{app:PauliPropExample}.

The procedure outlined above is a strategy for evaluating the objective function of the MPC optimization that avoids direct state vector simulation. However, it still incurs an exponential cost with the length of the prediction horizon, $N$. There are two types of exponential growth; first, the number of terms in the sum Eq. \eqref{eq:Pauliexp} can grow exponentially in $N$, and second, the weight of the Pauli terms in Eq. \eqref{eq:Pauliexp} can scale in the worst case as $O(2^N)$.
In order to mollify this cost, one can turn to truncation strategies. Two common truncation strategies are to remove terms from the expansion in Eq. \eqref{eq:Pauliexp} if (i) $|c_P(\Delta_t, \beta_k,...\beta_{k+N-1})| < \epsilon_{\rm th}$, or (ii) $w(P) > w_{\rm th}$, where $w(P)$ is the weight of the Pauli operator $P$ (number of non-identity terms). The thresholds $\epsilon_{\rm th}$ and $w_{\rm th}$ have to be chosen empirically. Carrying out such truncations can reduce the complexity of the predictive model and the number of Pauli operators to measure using classical shadows. However, the cost of this is a reduction in the accuracy of the model, which may need to be empirically evaluated. Future work can determine the effectiveness of truncations in enabling the predictive accuracy balanced with the computational tractability required for effective MPC.

We have outlined a reduced-order model that utilizes Pauli propagation with truncation and classical shadows to deliver a more efficient method for evaluating the MPC predictive model. The computational cost of this reduced-order model increases exponentially with the length of the predictive horizon, $N$, while we expect that the performance of MPC control policies to be more effective with increasing $N$. This presents a critical trade-off that must be negotiated for effective deployment of MPC-based quantum algorithms. The integration of high-performance computation (HPC) resources with quantum computers and the execution of the Pauli propagation reduced-order model on specialized HPC units (e.g., GPUs) is one way to extend the MPC prediction horizon.

\section{Numerical demonstrations}
\label{sec:numerics}
In this section, we present several numerical simulations demonstrating the ideas developed in previous sections and exploring the performance of the MPC-based quantum algorithm. We focus on two classes of problems widely studied in the optimization and quantum simulation literature. 

The first is a graph-based combinatorial optimization problem, \emph{Max-Cut}, which corresponds to finding a bipartition of a graph that maximizes the weight of the edges (or number of edges for unweighted graphs) connecting the partitions. The Max-Cut problem is an NP-hard combinatorial optimization problem that has been investigated extensively with classical algorithms~\cite{goemans1995improved} as well as with quantum algorithms such as QAOA~\cite{farhi2014quantumapproximateoptimizationalgorithm} and FALQON~\cite{magann2022feedback}. Quantum algorithms convert this optimization problem into an energy minimization task, where the solution of the problem is encoded into the ground state of a problem Hamiltonian $H_p$, defined on $n$ qubits for an $n$-node graph, as 
\begin{equation}
    H_p = \sum_{i,j\in\mathcal{E}}w_{ij}Z_iZ_j
    \label{eq:Hpdefmaxcut}
\end{equation}
where $w_{ij}$ are the edge weights and $\mathcal{E}$ is the edge set. For unweighted graphs $w_{ij}=1, \forall~i,j$.
In this work we use simulated instances of Max-Cut problems on graphs listed in Appendix~\ref{sec:MaxCutWeights}. For weighted graphs we use weights $w_{ij}$ sampled from a uniform random distribution between 0 and 2.

The second example we study is the transverse-field Ising model (TFIM), which is a model of nearest-neighbor interacting spins on a lattice subject to an external magnetic field. The task is to find the ground state energy and prepare the ground state of this model. The problem Hamiltonian in this case is given by
\begin{equation}
    H_p = J\sum_{\braket{i,j}}Z_iZ_j + h\sum_i X_i
    \label{eq:Hpdeftfim}
\end{equation}
where $\braket{i,j}$ indicates that the sum is over nearest neighbors on the lattice, $J$ is the coupling constant between spins, and $h$ is the transverse field strength. In the numerical analyses in this section we consider a $4 \times 2$ square lattice and use $J=1$ and $h=3$, chosen due to the existence of a quantum critical point near these parameter settings~\cite{PhysRevLett.110.135702}.

For both examples, we use the same driver Hamiltonian, 
\begin{equation}
    H_d = \sum_{j=1}^n X_j\label{eq:driver}
\end{equation}
for $n$ qubits, and the initial state is chosen to be $\ket{+}^{\otimes n}$.

To assess algorithm performance, we examine the values of $\langle H_p\rangle$ as well as the approximation ratio achieved after $k$ layers,
\begin{equation}
    r_{a,k} = \frac{\braket{H_p}_k}{\braket{H_{p}}_{min}},
\end{equation} 
where $\braket{H_{p}}_{min}$ is the minimum eigenvalue. An approximation ratio $r_{a,k} = 1$ indicates the state $\ket{\psi_k}$ is the ground state. In figures where logarithmic scales are used, the offset of the approximation ratio from its ground state value is used, 
\begin{equation}
    \chi_k = 1-r_{a,k}.
\end{equation}

We remark that the choices of $H_p$ in Eqs. \eqref{eq:Hpdefmaxcut} and \eqref{eq:Hpdeftfim} both admit an indefinite spectrum. Thus, for an arbitrary state, the quantity
$r_{a,k}$ could in principle be negative. In what follows, $r_{a,k}$ and $\chi_k$ are considered only for simulations that remain in the regime where $\braket{H_p}_k$ is negative and $r_{a,k}$ thus remains between 0 and 1. 

In the following numerical simulations of the MPC-based quantum algorithm, the L-BFGS-B algorithm~\cite{byrd1995limited} is used to solve the MPC optimization problem.  For $k=1$, $\beta_1^*$ is set to the FALQON gate parameter $-A_0$.  The MPC optimization problem is first solved for $k=2$, where the initial guess for each gate parameter in $\vec{\beta}_{(i|2)}$ is 0. At subsequent layers, the initial guess for each gate parameter in $\vec{\beta}_{(i|k)}$ is the optimal parameter determined at the previous layer, i.e., $\vec{\beta}_{(i|k)} = \vec{\beta}^*_{(i|k-1)}, i \in \{k, k+1, k+N-2\}$. The initial value for $\beta(k+N-1|k)$ is set to $\beta^*(k+N-2|k-1)$.  The gate parameter values are bounded between $[-2\pi, 2\pi]$.

\subsection{MPC Demonstration} \label{sec:MPCDemo}

We begin by substantiating the claim that the MPC-based quantum algorithm can achieve performance improvements over FALQON if properly designed.  We begin in Figure~\ref{fig:Weighted4NodeMaxcut} by comparing the performance of the MPC-based quantum algorithm using the terminal cost and cumulative cost objective functions of Eqs. \eqref{eq:TerminalCost} and~\eqref{eq:CumulativeStageCost}, respectively, against the performance of FALQON for solving the Max-Cut problem on weighted 4-node graphs.  The graph is shown in Appendix \ref{sec:MaxCutWeights}, along with the randomized edge weights for each of the ten instances considered. This simulation illustrates that the terminal cost MPC generally resulted in lower values of $\chi_k$ compared to both the cumulative cost MPC and FALQON. For each of the ten Max-Cut simulations, the terminal cost MPC produced a lower value of $\chi_k$ at some layer within the circuit than the lowest value of $\chi_k$ achieved under either the cumulative cost MPC or FALQON.  In only three of the ten Max-Cut instances did the cumulative cost MPC-based algorithm achieve a lower value of $\chi_k$ at some layer within the circuit than the lowest value of $\chi_k$ achieved under FALQON.  This indicates that proper objective function selection for the MPC-based algorithm can significantly impact whether the MPC can out-perform alternative ground state preparation algorithms. Furthermore, it establishes that it is possible to design the MPC-based quantum algorithm to outperform FALQON without additional structure to the optimization problem, such as the terminal constraint of Section~\ref{sec:Performance}.

\begin{figure}
    \centering
    \begin{overpic}[width = 1.08\columnwidth]{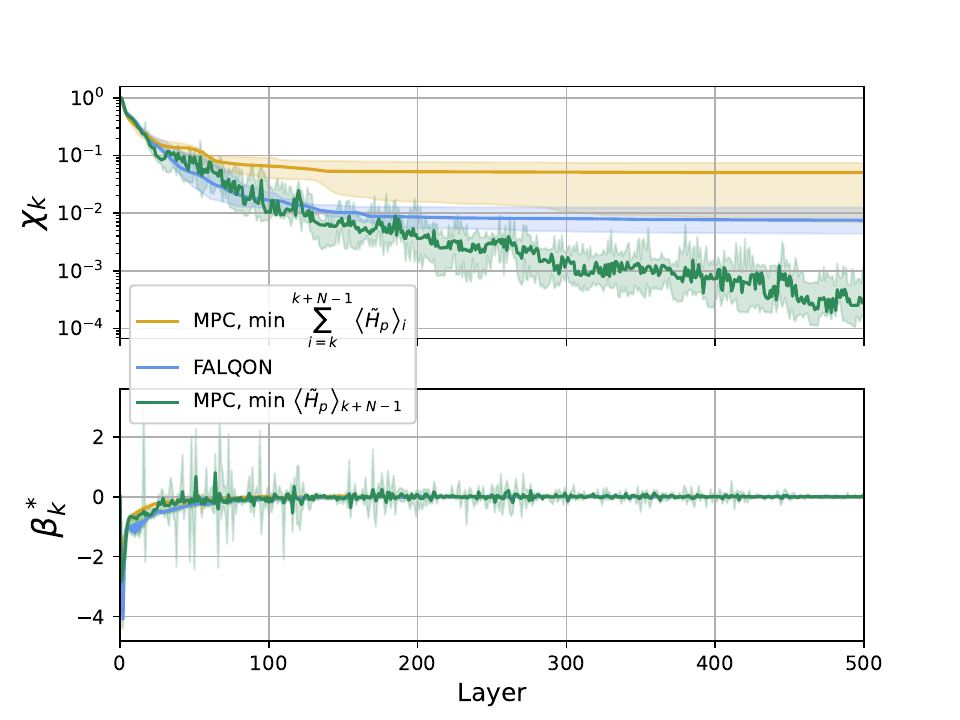}
        \put(2,68){\textbf{(a)}}
        \put(2,32){\textbf{(b)}}
    \end{overpic}
  \caption{\textbf{Performance of the terminal cost MPC and cumulative cost MPC compared to FALQON.} We simulate FALQON and the MPC-based algorithm for the Max-Cut problem formulated with the objective functions of Eqs.~\eqref{eq:TerminalCost}-\eqref{eq:CumulativeStageCost} over 10 instances of randomly weighted 4-node graphs. The median values of $\chi_k$ and of $\beta_k^*$ are plotted as solid lines in Panels (a) and (b), respectively, with the corresponding interquartile ranges shown in associated shaded regions of like colors. The step size for FALQON and for the MPC algorithm are set to $\Delta_t = 0.05$ and the prediction horizon for the MPC is set to $N = 10$.}
\label{fig:Weighted4NodeMaxcut}
\end{figure}

To provide further evidence of the potential of the MPC-based algorithm both with and without additional terminal constraints, we present a second example in Figure \ref{fig:LMPCMaxCut} that numerically demonstrates the performance of the MPC without terminal constraints compared to the terminal constraint MPC and FALQON.  This figure demonstrates the potential of an MPC without terminal constraints to achieve lower values of $\langle H_p \rangle$ than FALQON by the end of a fixed number of layers, demonstrates it may also achieve similar estimates of $\langle H_p \rangle_{\min}$ as FALQON when doing so, and confirms the role of the terminal constraint in grounding the MPC-based algorithm with the FALQON result.

\begin{figure*}
    \centering
    \begin{subfigure}{0.49\textwidth}
        \centering
        \begin{overpic}[width=1.08\columnwidth]{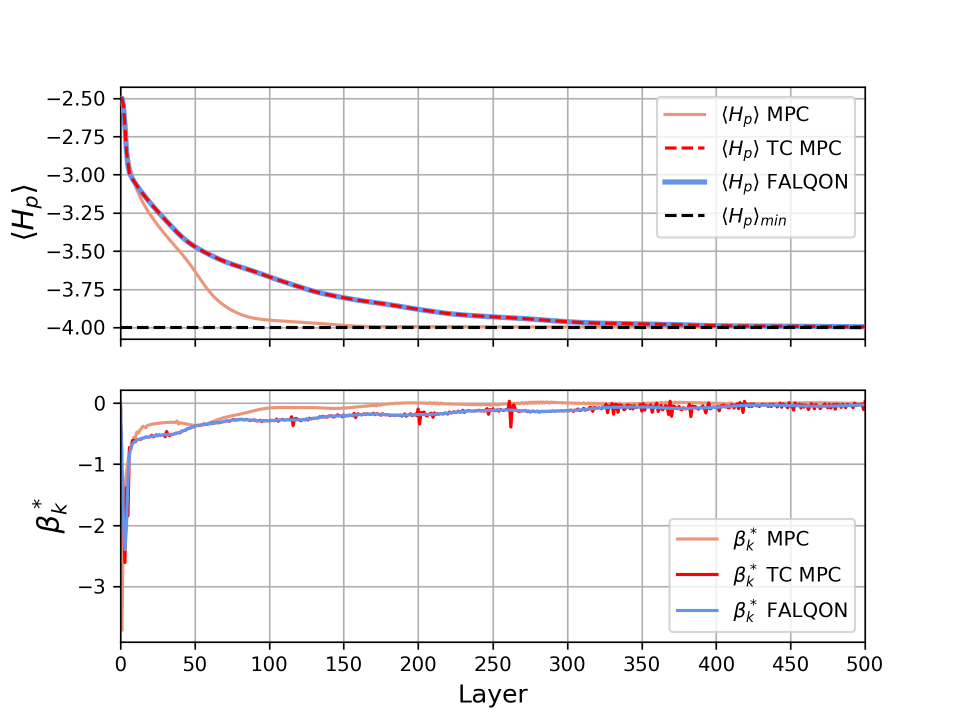}
            \put(2,68){\textbf{(a)}}
            \put(2,32){\textbf{(b)}}
        \end{overpic}
    \end{subfigure}
    \begin{subfigure}{0.49\textwidth}
        \centering
        \begin{overpic}[width=1.08\columnwidth]{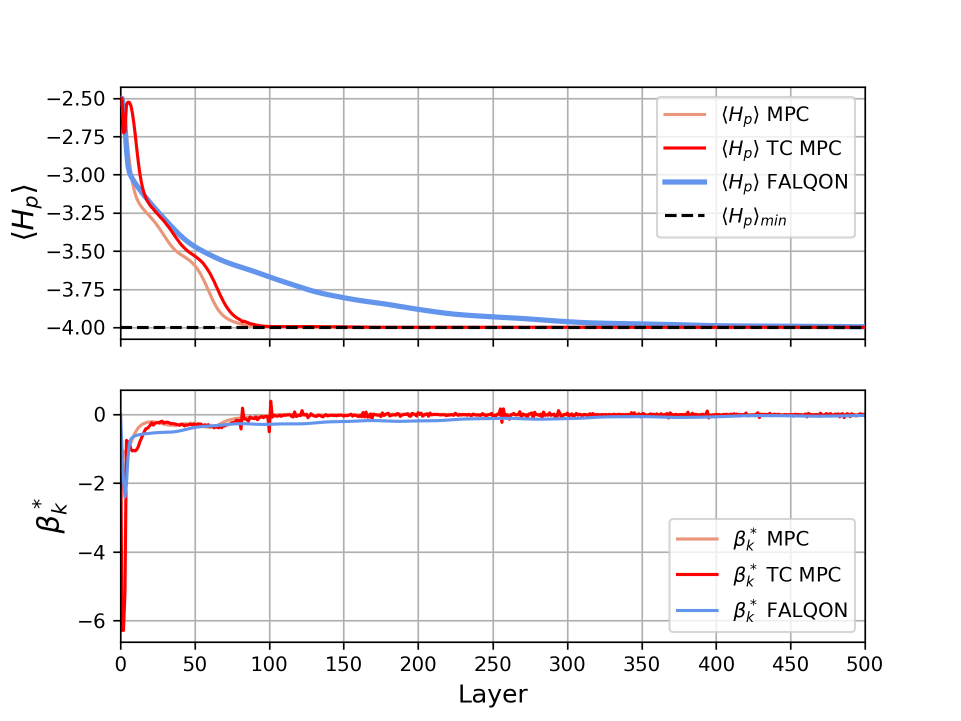}
            \put(2,68){\textbf{(c)}}
            \put(2,32){\textbf{(d)}}
        \end{overpic}
    \end{subfigure}
    \caption{\textbf{MPC-based quantum algorithm performance with $N=5$ and $N=400$.} We present two numerical examples comparing the performance of the MPC-based quantum algorithm using different prediction horizons against FALQON for solving Max-Cut on an unweighted 4-node graph. Results are plotted for both the base-case MPC of Section~\ref{sec:MPCAlgorithm} with the objective function of Eq.~\eqref{eq:LMPCAlgorithm:Objective} as well as the terminal constraint MPC. All results were obtained using a step size of $\Delta_t = 0.05$. In Panel (a), the base-case MPC (in orange) and the terminal constraint MPC (the dashed red line) both use a prediction horizon of $N=5$. Here, we see that when using a short prediction horizon, the trajectory of $\langle H_p \rangle$ under the terminal constraint MPC closely follows that under the FALQON trajectory, shown in blue. In Panel (c), when the prediction horizon is increased to $N=400$ in both the base-case MPC and the terminal constraint MPC, the trajectory of $\langle H_p \rangle$ under the terminal constraint MPC more closely resembles that under the base-case MPC trajectory. The gate parameters $\beta^*_k$ for $N=5$ are plotted in Panel (b) and for $N=400$ in Panel (d). The intent of the constraint of Eq.~\eqref{eq:LMPCAlgorithm:TerminalConstraint} was enforced in these simulations via an inequality constraint of the form $0 \leq |1-|\braket{\tilde{\psi}_{k+N-1}|\phi_{k+N-1}}|^2| \leq \epsilon$ where $\epsilon = 10^{-15}$.}
    \label{fig:LMPCMaxCut}
\end{figure*}

\subsection{Analysis of hyperparameters} \label{sec:Hyperparameters}
In this section we analyze the behavior of the MPC-based quantum algorithm as a function of the key hyperparameters $\Delta_t$ and $N$. In order to isolate the effects of these hyperparameters, we limit ourselves to small problem sizes and use full state vector simulation to evaluate the MPC model, \emph{i.e.,} no reduced-order model is used.  

The product $N\Delta_t$ represents the prediction time, and intuitively, larger values allow for longer-term dynamics to be accounted for in the prediction horizon and discourage short-sighted decisions when determining optimal gate parameters. For fixed $\Delta_t$, the larger $N$ is, the further into the future the MPC-based algorithm will predict the dynamics and the greater the number of gate parameters to optimize over will be. Thus, large prediction horizons can increase the flexibility in finding a set of parameterized gates that may produce qubit states with smaller values of $\braket{H_p}$. To demonstrate this, Figure~\ref{fig:4NodeVaryingHorizons} shows the value of $\chi_k$ evaluated at the end of each of 500 layers for an unweighted 4-node Max-Cut problem using the terminal cost MPC, for values of $N$ between 2 and 60.  For $N=2$, the MPC-based algorithm decreases $\chi_k$ in the first few layers more rapidly than it does when $N=60$, but by layer 500, the value of $\chi_k$ is suboptimal relative to larger values of $N$.  This can be understood as an initially \textit{greedy} optimization, where the MPC-based algorithm with the short horizon lacks the ability for long-term planning, potentially sacrificing long-term optimality for short-term improvements. Increases in $N$ generally encourage lower values of $\chi_k$ to be found, and in a reduced number of layers compared to smaller prediction horizons, however this trend is not strictly followed as $N$ increases. All of the trajectories of $\chi_k$ in Figure~\ref{fig:4NodeVaryingHorizons} exhibit fluctuations that increase the value of $\chi_k$ temporarily to then subsequently decrease it further.  These fluctuations can be considered to be a feature of the terminal cost MPC, in that they can enable improved performance relative to a baseline algorithm such as FALQON. For further discussion of this point, see Appendix~\ref{sec:FluctuationAnalysis}.

\begin{figure}
        \centering
        \includegraphics[width=1.08\columnwidth]{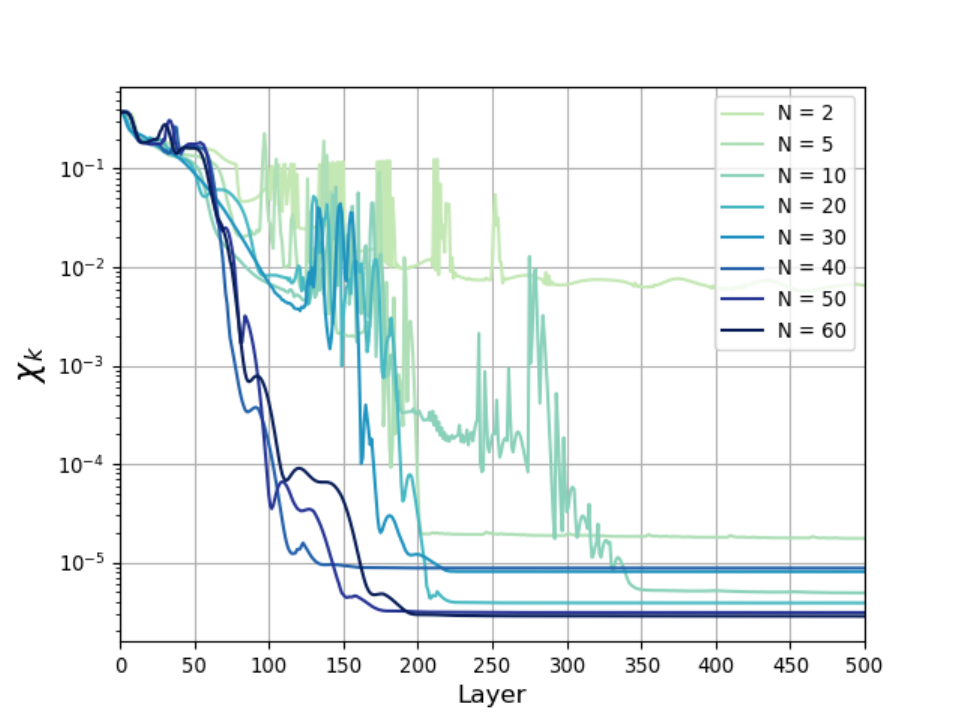}
        \caption{\textbf{Effect of varying prediction horizon.} We simulate the performance of the terminal cost MPC on an unweighted 4-node Max-Cut problem with varying prediction horizons and $\Delta_t = 0.05$ over 500 layers. Larger prediction horizons, shown in darker colors, tend to encourage more optimal values of $\chi_k$ achieved with shorter circuits relative to shorter horizons.}
        \label{fig:4NodeVaryingHorizons}
\end{figure}

We analyze the combined dependence of the terminal cost MPC on $\Delta_t$ and $N$ in Figure~\ref{fig:Heatmaps}. We observe that larger values of $N$ generally enable both higher values of the approximation ratio to be achieved across a wider range of $\Delta_t$ values, per Panel (a), as well as  approximation ratios of at least 0.97 to be achieved with shorter circuits, per Panel (b). This motivates a possible heuristic for hyperparameter selection: first, determine the largest value of $N$ for use in the MPC, depending on available computational budget; then, for this maximal allowable value of $N$, scan over a range of possible $\Delta_t$ to determine the value that produces best performance.

\begin{figure*}[t]
    \begin{subfigure}{0.49\textwidth}
        \centering
        \begin{overpic}[width=1.08\columnwidth]{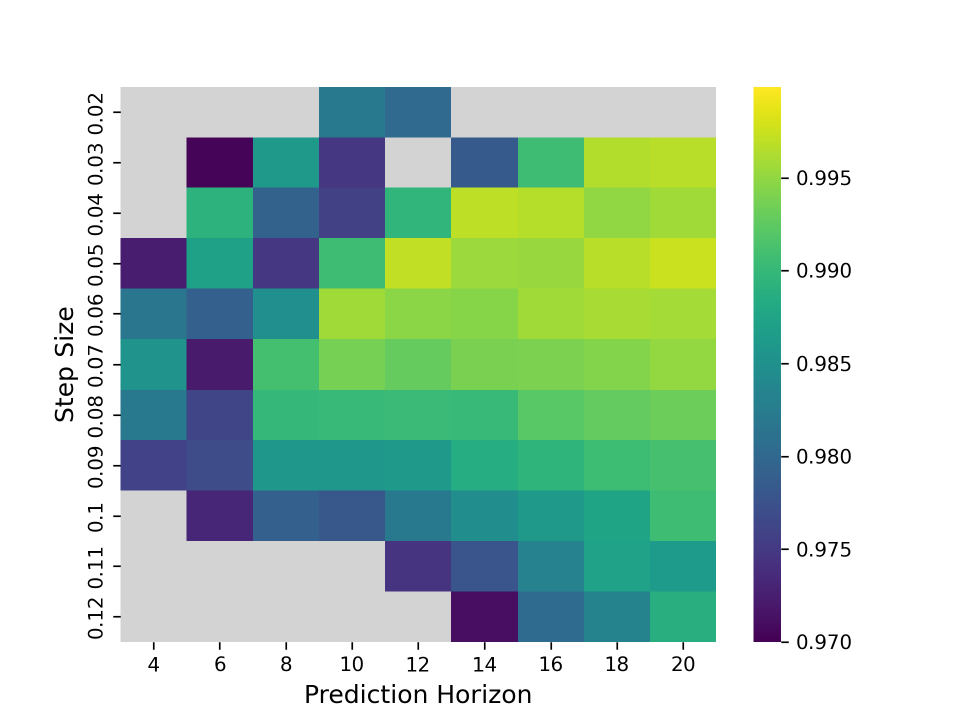}
            \put(0,65){\textbf{(a)}}
        \end{overpic}
    \end{subfigure}
    \begin{subfigure}{0.49\textwidth}
        \centering
        \begin{overpic}[width=1.08\columnwidth]{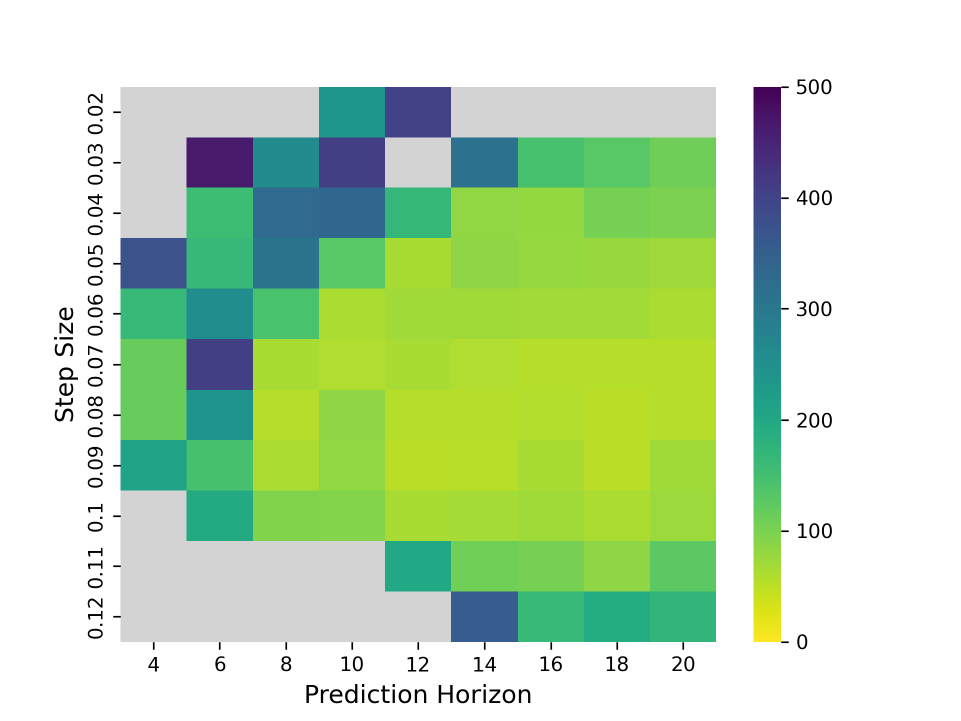}
            \put(0,65){\textbf{(b)}}
        \end{overpic}
    \end{subfigure}
    \caption{\textbf{Combined effect of prediction horizon and step size.} We simulate the terminal cost MPC with varying prediction horizons and step sizes over 500 layers for a solving a Max-Cut problem on the weighted, 6-node graph shown in Appendix~\ref{sec:MaxCutWeights}. Panel (a) shows the maximum approximation ratio achieved over 500 layers for each combination of $N$ and $\Delta_t$. Panel (b) shows the first layer with an approximation ratio of at least $0.97$.  Grey boxes indicate that an approximation ratio of 0.97 was not achieved.  Due to the fluctuations that can occur in $r_{a,k}$ between layers under the terminal cost MPC, as shown in Figure~\ref{fig:4NodeVaryingHorizons}, the maximum approximation ratio may not be achieved at the final layer, and the approximation ratio may decrease below 0.97 after the layer marked in Panel (b).}
    \label{fig:Heatmaps}
\end{figure*}

\begin{figure}
    \centering
    \includegraphics[width=1.08\columnwidth]{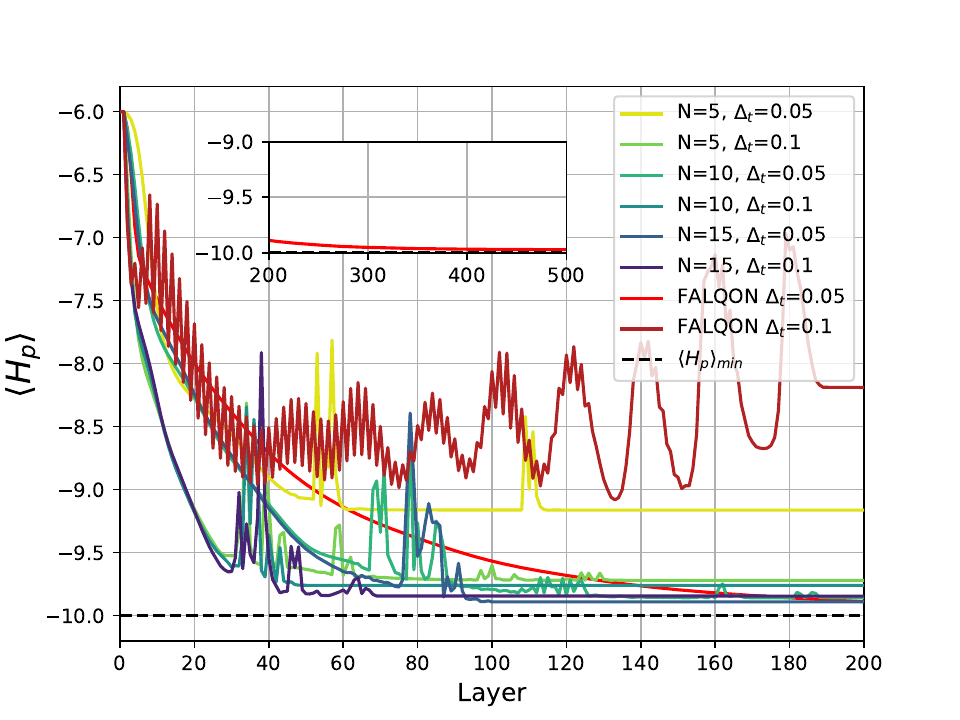}
    \caption{\textbf{MPC-based quantum algorithm with state vector simulation for 8-node Max-Cut.} We simulate the terminal cost MPC over 200 layers with state vector simulation using a range of hyperparameters $N$ and $\Delta_t$ and compare the performance with two instances of FALQON where $\Delta_t = 0.05$ and 0.1. In every instance except where $N=5$ and $\Delta_t = 0.05$, the terminal cost MPC converges to an estimate of $\braket{H_p}$ in a reduced number of layers compared to FALQON. However, the inset showing $\braket{H_p}$ over layers 200 through 500 under FALQON with $\Delta_t = 0.05$ illustrates that FALQON eventually converges closer to the ground state compared to the terminal cost MPC shown here. The oscillatory behavior of FALQON when $\Delta_t = 0.1$ is due to this value of $\Delta_t$ exceeding the critical $\Delta_t$ value that ensures monotonic descent \cite{magann2022feedback}.  The MPC that use this value of $\Delta_t$ overcome this failure mode of FALQON due to the receding horizon optimization, and still can obtain reasonable performance at this larger value of $\Delta_t$.}
    \label{fig:8NodeVaryingN}
\end{figure}

\begin{figure*}[t]
    \centering
    \begin{subfigure}{0.49\textwidth}
        \centering
        \begin{overpic}[width=1.08\textwidth]{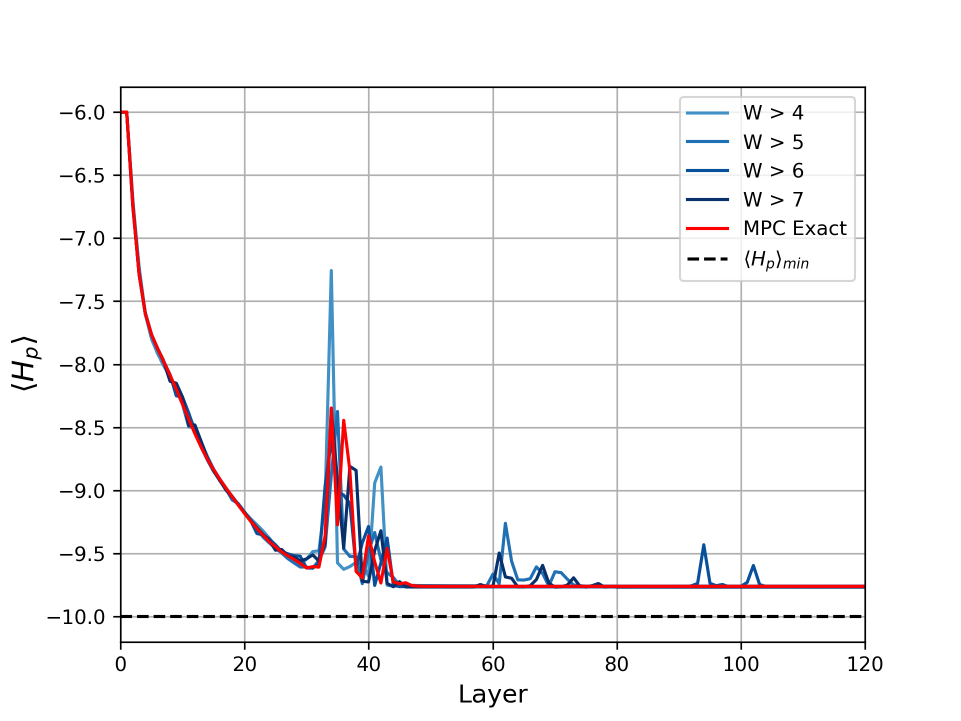} 
            \put(0,65){\textbf{(a)}}
        \end{overpic}
        \label{fig:8NodeWeights}
    \end{subfigure}
    \begin{subfigure}{0.49\textwidth}
        \centering
        \begin{overpic}[width=1.08\textwidth]{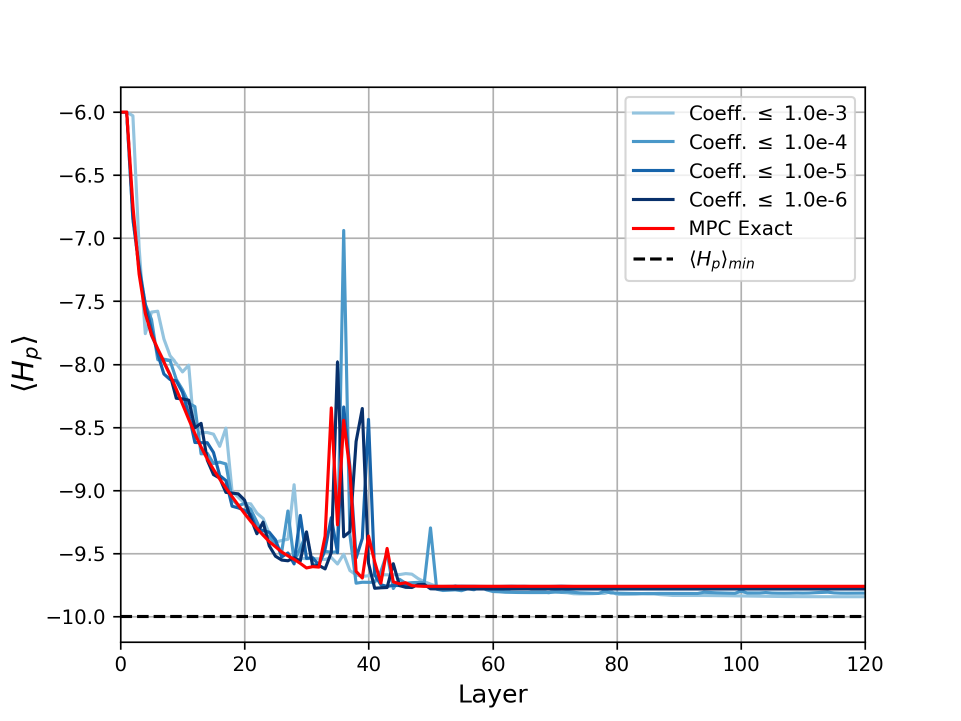}
            \put(0,65){\textbf{(b)}}
        \end{overpic}
        \label{fig:8NodeFreqs}
    \end{subfigure}
\caption{\textbf{MPC-based quantum algorithm with Pauli Propagation for 8-node Max-Cut.} We simulate the terminal cost MPC with Pauli propagation using varying weight-based and coefficient-based truncations.  All simulations use $\Delta_t = 0.1$ and $N=10$.  ``MPC Exact'' represents the trajectory of $\langle H_p \rangle$ when the MPC-based quantum algorithm employs state vector simulation in computing $\langle \tilde{\psi}_{k+N-1}| H_p | \tilde{\psi}_{k+N-1}\rangle$, and $\langle H_p \rangle_{\min}$ reflects the ground state eigenvalue of $H_p$. Panel (a) showcases the trajectories of $\langle H_p \rangle$, where $W > q$ signifies that Pauli operators with weights greater than $q$, $q = 4, 5, 6$, or 7, are truncated. Panel (b) showcases the trajectories of $\langle H_p \rangle$, where Coeff. $\leq q$ signifies that Pauli operators with coefficients having a modulus below $q = 10^{-3}, 10^{-4}, 10^{-5}$, and $10^{-6}$ are truncated.}
    \label{fig:8NodeSims}
\end{figure*}

\begin{figure}[t]
    \centering
    \includegraphics[width=1.08\columnwidth]{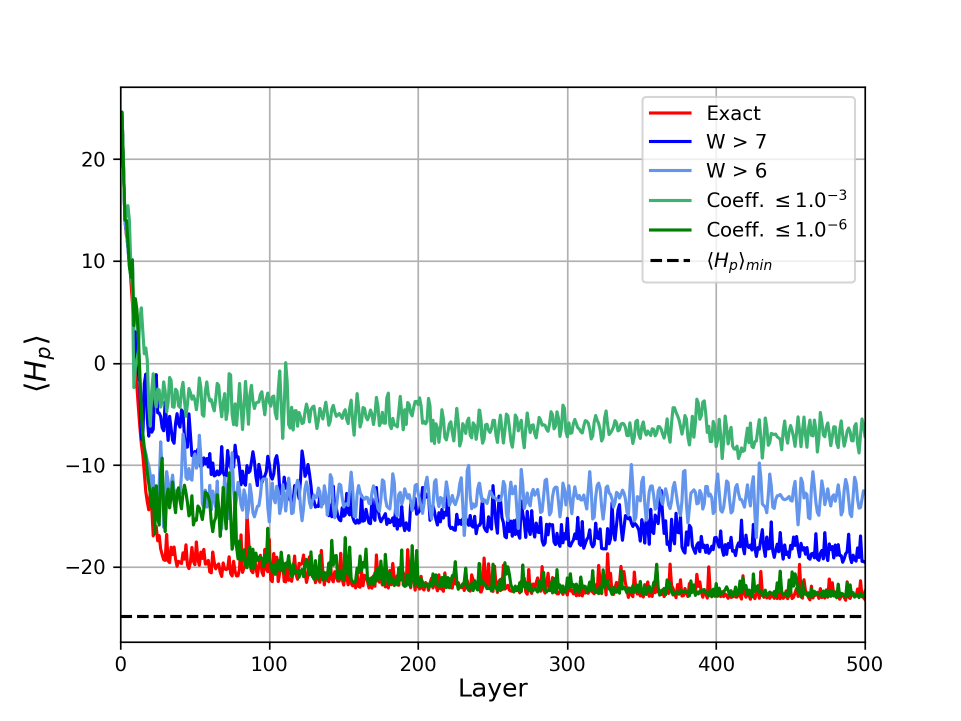}
    \caption{\textbf{Transverse field Ising model with Pauli propagation.}  We compare the performance of the terminal cost MPC using Pauli propagation with various weight-based and frequency-based truncations used to obtain $\tilde{\braket{H_p}}_{k+N-1}$ against the trajectory of the MPC-based quantum algorithm using state vector simulation, where in each case, $\Delta_t = 0.1$ and $N = 10$.  Weight-based truncations are shown in blue, where Pauli operators with weight greater than 6 and 7 (labeled $W > 6$ and $W > 7$), respectively, are trimmed. Coefficient-based truncations are shown in green, where Pauli operators with coefficients less than $10^{-3}$ and $10^{-6}$ (labeled Coeff. $\leq 10^{-3}$ and Coeff. $\leq 10^{-6}$) are trimmed, respectively. The ground state eigenvalue is denoted by $\langle H_p \rangle_{\min}$ and represented by a dashed black line.}
    \label{fig:TFIMPauliProp}
\end{figure}

\subsection{Max-Cut and TFIM case studies}
\label{sec:maxcut_tfim}
In this section, we present numerical demonstrations of the MPC-based quantum algorithm on the Max-Cut and TFIM example models discussed above. We utilize the Pauli propagation reduced-order model with weight-based and coefficient-based truncations, as discussed in Sec. \ref{sec:PauliPropMPC}, to simulate the MPC model. We utilize PauliPropagation.jl~\cite{rudolph2025pauli,rudolphSoftware} and  Optimization.jl packages for these simulations, with Nelder-Mead~\cite{nelder1965simplex} used as the optimization solver.

\subsubsection{Max-Cut}
Figure~\ref{fig:8NodeVaryingN} shows the result of applying terminal cost MPC using exact state vector simulation with a range of $N$ and $\Delta_t$ to solve Max-Cut on an 8-node graph.  The figure indicates that the MPC hyperparameters impact the degree of offset of the values of $\langle H_p \rangle$ from $\langle H_p \rangle_{\min}$ after 200 layers.  We conclude that a reasonable estimate of the ground state can be obtained from the MPC in less layers than FALQON, though FALQON may continue to improve its estimate after the 200 layers over which the MPC was tested.  Figure~\ref{fig:8NodeSims} shows the result of applying the terminal cost MPC using Pauli propagation to the same 8-node graph, where the results are compared to the trajectory in Figure~\ref{fig:8NodeVaryingN} where $N=10$ and $\Delta_t = 0.1$. In Panel (a), the trajectories that use the various weight-based truncations each reach similar values of $\langle H_p \rangle$ in a similar number of layers compared to each other and to the MPC with perfect state predictions.  A similar effect is also shown in Panel (b) when coefficient-based truncations are used, however, slightly lower values of $\langle H_p \rangle$ are found by the end of 120 layers compared to when state vector simulation was used in the MPC.  This occurs because the MPC is optimizing an objective function that depends on a model of the dynamics that is potentially different from the dynamics when the circuit is applied. This, combined with the selection of $N$ and $\Delta_t$, in this case results in improved performance compared to the use of state vector simulation, though in general we would not expect the use of a reduced-order model to improve performance of the MPC. 
The use of truncated Pauli propagation models was thus effective for performing approximate ground state preparation, demonstrating that the MPC-based algorithm can achieve good performance with reduced-order models. 

\subsubsection{Transverse-Field Ising Model}\label{sec:PauliPropTFIM}
Figure~\ref{fig:TFIMPauliProp} shows the results of applying terminal cost MPC to the $4\times 2$ TFIM. Both of the weight-based truncations failed to prepare states with values of $\langle H_p \rangle$ as low as those prepared when using a perfect model, indicating that weight-based truncations may not be suitable for all problem instances. Coefficient-based truncations that trimmed Pauli operators with a coefficient less than $10^{-3}$ also reduced the performance of the algorithm compared to using a perfect model, but with the cut-off of $10^{-6}$, the performance was more similar to using state vector simulation in the MPC.  This highlights the importance of selecting an adequate reduced-order model for pairing with the MPC-based algorithm for a given problem, and demonstrates that using coefficient-based truncations may be a suitable approach to controlling the accuracy of Pauli propagation when weight-based truncations result in poor performance.   These results provide evidence that the MPC-based algorithm can be used with reduced-order modeling strategies to estimate $\langle H_p \rangle_{\min}$ across different computational problems. They also indicate that there are key cost/accuracy tradeoffs between using weight-based and coefficient-based truncations, and we leave a systematic investigation of these tradeoffs to future work.  The results are consistent with the well-established principles in MPC that the controller may reach a similar target to that achieved with a perfect system model if the model is not perfect but is still reasonably accurate, though potentially with off-set from a target.

\section{Conclusions}
In this work, we developed a hybrid quantum-classical algorithm based on model predictive control, expanding the toolbox of variational quantum algorithms. The algorithm is equipped with a classical component used to simulate the quantum system over finite prediction horizons that can be selected to enable tractable optimization of circuit parameters. We demonstrated that with careful modification of the structure of the algorithm, it can be theoretically guaranteed to perform at least as well as a feedback-based quantum algorithm. With the potential for performance improvements established, we relaxed the structure of the algorithm and investigated the effects of the various design choices on the performance of the algorithm, including the objective function used in the optimization problem, hyperparameters such as prediction horizon length and step size, and choice of predictive model. Through mathematical analyses and simulations, we demonstrated how these design choices for the MPC-based algorithm interact with one another, and provided suggestions for tuning the algorithm to attempt to trade off between performance and classical and quantum computing resource use. In these demonstrations, we established that reduced-order modeling strategies can be effectively implemented in the algorithm to allow for tractable implementations.

The MPC-based algorithm can be designed to recover both FALQON and QAOA trajectories. Specifically, in the case that the terminal constraint MPC of Section~\ref{sec:Performance} is operated with $N = 1$, the gate parameterizations are given by FALQON.  QAOA can be recovered using the same circuit structure as QAOA, a prediction horizon of the same length as the number of layers used in QAOA, a perfect model of the quantum circuit, and globally optimal gate parameters from QAOA and MPC.  Then, the "open-loop" parameterization under the terminal cost MPC solved at $k=1$, meaning the solution that would be applied if all of the gate parameters computed when $k=1$ are implemented in the MPC-based algorithm circuit instead of only the first, gives the same parameterizations as under QAOA.   The closed-loop solution with a shrinking horizon starting from $k=1$ will also give the QAOA parameterization.  Thus, the MPC-based algorithm can be seen as a generalization of FALQON and QAOA that can be tuned to emphasize the strengths of either for specific problems. Future work can further examine performance comparisons with QAOA implementations. 

The set of algorithm design aspects that we investigated does not cover every possible design decision that may be made for the algorithm.  For example, the optimization problem of Eq. \eqref{eq:MPCAlgorithm} is solved with a classical optimizer, meaning that the selection of the classical optimization strategy and whether it performs local or global optimization, as well as the initial guess of the decision variables in the algorithm, can impact both the classical computing time required, as well as the solution quality, e.g., how closely it approximates the minimum value of $\langle H_p \rangle$. Future work can also further investigate tuning guidelines for setting up the MPC-based quantum algorithm.

\section*{Acknowledgements}

This material is based upon work supported by the U.S. Department of Energy, Office of Science, Office of Workforce Development for Teachers and Scientists, Office of Science Graduate Student Research (SCGSR) program. The SCGSR program is administered by the Oak Ridge Institute for Science and Education (ORISE) for the DOE. ORISE is managed under ORAU contract number DE-SC0014664. All opinions expressed in this paper are the author's and do not necessarily reflect the policies and views of DOE, ORAU, or ORISE. This work was supported by the Laboratory Directed Research and Development program (Project 233972) at Sandia National Laboratories, a multimission laboratory managed and operated by National Technology and Engineering Solutions of Sandia LLC, a wholly owned subsidiary of Honeywell International Inc. for the U.S. Department of Energy’s National Nuclear Security Administration under contract DENA0003525. This paper describes objective technical results and analysis. Any subjective views or opinions that might be expressed in the paper do not necessarily represent the views of the U.S. Department of Energy or the United States Government. This article has been authored by an employee of National Technology \& Engineering Solutions of Sandia, LLC under Contract No. DE-NA0003525 with the U.S. Department of Energy (DOE). The employee owns all right, title and interest in and to the article and is solely responsible for its contents. The United States Government retains and the publisher, by accepting the article for publication, acknowledges that the United States Government retains a non-exclusive, paidup, irrevocable, world-wide license to publish or reproduce the published form of this article or allow others to do so, for United States Government purposes. The DOE will provide public access to these results of federally sponsored research in accordance with the DOE Public Access Plan https://www.energy.gov/downloads/doe-public-access-plan. SAND2026-23216O. Helen Durand was employed in Summer 2024 at Sandia National Laboratories.  Financial support from Wayne State University and NSF CBET-2143469 is also gratefully acknowledged.

\bibliographystyle{IeeeTran.bst}
\bibliography{apssamp}% Produces the bibliography via BibTeX.

\appendix

\section{}

These appendices expand on the concepts and simulations presented in the main text and provide information relevant to the implementation of the MPC-based quantum algorithm.  
Appendix~\ref{app:LMPCProof} provides a proof of Theorem~\ref{thm:terminalconstraintLMPC}.  An example of Pauli propagation and the truncation strategies relevant to the MPC-based algorithm is discussed in Appendix~\ref{app:PauliPropExample}.  Appendix~\ref{sec:MaxCutWeights} provides the weights for the weighted 4-node Max-Cut simulations of Figure~\ref{fig:Weighted4NodeMaxcut} as well as an explicit form of the problem Hamiltonian used in the 6-node Max-Cut simulations of Figure~\ref{fig:Heatmaps}. In Appendix~\ref{sec:FluctuationAnalysis}, we discuss the fluctuations in the value of $\langle H_p \rangle$ across a circuit as observed under the terminal cost MPC and noted in the main text.

\subsection{Terminal Constraint MPC-based Algorithm Performance Proof}\label{app:LMPCProof}

\noindent \textbf{Proof.} This proof consists of two parts, following the closed-loop terminal constraint MPC performance proof of~\cite{ellis2016closed}. In Part 1, we prove the following bound holds for the MPC-based algorithm using a fixed prediction horizon:
\begin{equation}
    \sum_{k=1}^{f} \braket{\psi_{k}|H_p|\psi_{k}}_{\beta^*_k} \leq \sum_{k=1}^{f+N} \braket{\phi_{k}|H_p|\phi_{k}}_{\nu_k}\label{eq:prop1}
\end{equation}
In Part 2, we prove the bound of Eq. \eqref{eq:prop1Option2} holds when using a shrinking prediction horizon for the last $N$ layers before layer $f$.

\textit{Part 1.}
At layer $k$, the parameterization given by FALQON, $\beta(i|k) = \nu_i$ for $i \in \{k,\dots,k+N-1\}$ is a feasible solution to the optimization problem of Eq. \eqref{eq:LMPCAlgorithm} satisfying the constraint of Eq. \eqref{eq:LMPCAlgorithm:TerminalConstraint}. Let $\beta^*(i|k)$ for $i \in \{ k,\dots,k+N-1 \}$ be the optimal gate parameters computed at layer $k$ for the MPC-based algorithm in Eq. \eqref{eq:LMPCAlgorithm}. The optimal value of the objective function of Eq. \eqref{eq:LMPCAlgorithm:Objective} at layer $k$, denoted $J_k^*$, can be written as follows:
\begin{equation}
    J^*_{k} = \sum_{i=k}^{k+N-1}\braket{\tilde{\psi}_{i}|H_p|\tilde{\psi}_{i}}_{\beta^*(i|k)}\label{eq:k}
\end{equation}
where $\sum_{i=k}^{k+N-1}\braket{\tilde{\psi}_{i}|H_p|\tilde{\psi}_{i}}_{\beta^*(i|k)}$ signifies that values of $\ket{\tilde{\psi}_i}$, $i=k,\ldots,k+N-1$, used in evaluating the expectation value of $H_p$ at the end of each layer to form this sum come from state predictions under the gate parameter trajectory $\beta^*(i|k)$, $i=k,\ldots,k+N-1$, initialized from the state at the start of the $k$-th layer.  At layer $k+1$, the optimal parameters calculated at the previous layer, $\beta(i|k+1) = \beta^*(i|k)$ for $i \in \{k+1,\dots,k+N-1\}$ and the parameter prescribed by FALQON at layer $k+N$, $\beta(k+N|k+1) = \nu_{k+N}$, is a feasible solution to Eq. \eqref{eq:LMPCAlgorithm}. The optimal value of Eq. \eqref{eq:LMPCAlgorithm} $J^*_{k+1}$ is upper bounded by the value of Eq. \eqref{eq:LMPCAlgorithm} under a feasible solution such that:
\begin{equation}
\begin{aligned}
    J^*_{k+1} & \leq \sum_{i=k+1}^{k+N-1}\braket{\tilde{\psi}_{i}|H_p|\tilde{\psi}_{i}}_{\beta^*(i|k)} \\
    & + \braket{\phi_{k+N}|H_p|\phi_{k+N}}_{\nu_{k+N}}\label{eq:kplus1}
\end{aligned}
\end{equation}
Subtracting Eq. \eqref{eq:k} from Eq. \eqref{eq:kplus1}, the difference between the optimal values of the objective function in any two consecutive layers can be bounded as follows:
\begin{equation}
\begin{aligned}
    J_{k+1}^*-J_k^* & \leq \braket{\phi_{k+N}|H_p|\phi_{k+N}}_{\nu_{k+N}} \\
    & - \braket{\tilde{\psi}_{k}|H_p|\tilde{\psi}_{k}}_{\beta^*_k}\label{eq:difference}
\end{aligned}
\end{equation}
Let $f$ denote any positive integer representing the total number of layers. Taking the sum of the differences between the optimal values at any two consecutive layers given by Eq. \eqref{eq:difference} over the total number of layers, we obtain the following:
\begin{equation}
\begin{aligned}
    & \sum_{k=1}^{f}[J_{k+1}^*-J_k^*] \\
    & \leq \sum_{k=N+1}^{f+N}\braket{\phi_{k}|H_p|\phi_{k}}_{\nu_k} - \sum_{k=1}^{f}\braket{\tilde{\psi}_{k}|H_p|\tilde{\psi}_{k}}_{\beta^*_k} \\
    & = \sum_{k=N+1}^{f+N}\braket{\phi_{k}|H_p|\phi_{k}}_{\nu_k} - \sum_{k=1}^{f}\braket{\psi_{k}|H_p|\psi_{k}}_{\beta^*_k} \label{eq:sumofdiff}
\end{aligned}
\end{equation}
where the last equality follows from the fact that the final summation only includes terms associated with the state predictions under the implemented gate parameters, which in the perfect model case considered, correspond to the actual state trajectory under those gate parameters. 
The inclusion of a hypothetical $J_{f+1}^*$ corresponding to an optimal objective function value beyond the final layer in the summation of Eq.~\eqref{eq:sumofdiff} could be avoided by modifying the summation to only consider up to layer $f-1$, with the corresponding modifications to the subsequent steps and bounds on the summations in Eq. \eqref{eq:prop1}, but this will be less useful for Part 2 of the proof, so we retain $J_{f+1}^*$ in this proof.

Assuming, without loss of generality, that $\langle \psi|H_p|\psi \rangle \geq 0$, for all $\ket{\psi}$ in the $2^n$-dimensional complex Hilbert space containing the quantum state, the left hand side of Eq. \eqref{eq:sumofdiff} is bounded below as follows:
\begin{equation}
\begin{aligned}
    \sum_{k=1}^{f}[J_{k+1}^*-J_k^*] & = J_{f+1}^*-J_{1}^* \\
    & \geq -J_{1}^*\label{eq:lowerbound} 
\end{aligned}
\end{equation}
Furthermore, because the gate parameters determined by FALQON would be feasible for the MPC of Eq. \eqref{eq:LMPCAlgorithm}, the objective function under the optimal gate parameters $\beta^*(i|k)$ for $i\in \{k,\ldots,k+N-1\}$ for $k=1$ is upper bounded by the objective function value under the feasible FALQON gate parameters as follows:
\begin{equation}
    J_{1}^* \leq \sum_{k=1}^{N}\braket{\phi_{k}|H_p|\phi_{k}}_{\nu_k}\label{eq:upperbound}
\end{equation}
$\sum_{k=1}^{f}[J_{k+1}^*-J_k^*]$ is then upper bounded by the right-hand side of Eq. \eqref{eq:sumofdiff} and lower bounded by the right-hand side of Eq. \eqref{eq:lowerbound}.  Therefore:
\begin{equation}
\begin{aligned}
    -J_{1}^* \leq \sum_{k=N+1}^{f+N}\braket{\phi_{k}|H_p|\phi_{k}}_{\nu_k} - \sum_{k=1}^{f}\braket{\psi_{k}|H_p|\psi_{k}}_{\beta^*_k} 
\end{aligned}
\end{equation}
which can be rewritten as:
\begin{equation}
\begin{aligned}
    -\sum_{k=N+1}^{f+N}\braket{\phi_{k}|H_p|\phi_{k}}_{\nu_k} + \sum_{k=1}^{f}\braket{\psi_{k}|H_p|\psi_{k}}_{\beta^*_k} \leq J_{1}^* \label{eq:upperbound2}
\end{aligned}
\end{equation}
Combining the right-hand side of Eq. \eqref{eq:upperbound} and left-hand side of Eq. \eqref{eq:upperbound2} gives:
\begin{equation}
\begin{aligned}
    -\sum_{k=N+1}^{f+N}\braket{\phi_{k}|H_p|\phi_{k}}_{\nu_k} & + \sum_{k=1}^{f}\braket{\psi_{k}|H_p|\psi_{k}}_{\beta^*_k}\\
    & \leq \sum_{k=1}^{N}\braket{\phi_{k}|H_p|\phi_{k}}_{\nu_k}
\end{aligned}
\end{equation}
Adding $\sum_{k=N+1}^{f+N}\braket{\phi_{k}|H_p|\phi_{k}}_{\nu_k}$ to both sides, the sum of Eq. \eqref{eq:LMPCAlgorithm:Objective} from layer $1$ to layer $f$ under the MPC-based algorithm of Eq. \eqref{eq:LMPCAlgorithm} is no greater than the sum of Eq. \eqref{eq:LMPCAlgorithm:Objective} from layer $1$ to layer $f+N$ under FALQON, proving the bound of Eq. \eqref{eq:prop1}.

\textit{Part 2.}  We now prove that Eq. \eqref{eq:prop1Option2} holds when a shrinking horizon is used for the last $N$ sampling periods before time step $f$.  In this case, the sum of the difference of the optimal solutions to Eq. \eqref{eq:LMPCAlgorithm:Objective} between consecutive layers for $k \in \{1,\ldots,f-N\}$ before a shrinking horizon is used, given by Eq.~\eqref{eq:sumofdiff}, is as follows:
\begin{equation}
    \begin{aligned}
        & \sum_{k=1}^{f-N} [J_{k+1}^* - J_k^*] \\
        & \leq \sum_{k=N+1}^{f}\braket{\phi_{k}|H_p|\phi_{k}}_{\nu_k} - \sum_{k=1}^{f-N}\braket{\psi_{k}|H_p|\psi_{k}}_{\beta^*_k}\label{eq:DiffBeforeShrinkingHorizon}
    \end{aligned}
\end{equation}
Using similar steps as in Part 1 starting from Eq. \eqref{eq:lowerbound}, we conclude that the following holds (a version of Eq.~\eqref{eq:prop1}):
\begin{equation}
    \begin{aligned}
        \sum_{k=1}^{f-N} \braket{\psi_{k}|H_p|\psi_{k}}_{\beta^*_k} \leq \sum_{k=1}^{f} \braket{\phi_{k}|H_p|\phi_{k}}_{\nu_k}
    \end{aligned}
\end{equation}
We now consider the last $N$ sampling periods, over which a shrinking horizon is employed.  In this case, we denote the optimal solution at layer $k \in \{f-N+1,\ldots,f\}$ as:
\begin{equation}
    \begin{aligned}
        J_k^* = \sum_{i=k}^{k+N_k-1}\braket{\tilde{\psi}_{i}|H_p|\tilde{\psi}_{i}}_{\beta^*(i|k)}\label{eq:optsln_k}
    \end{aligned}
\end{equation}
where $N_k = N - j + 1$ and $j = k - f + N$ for $k \in \{f-N+1,\dots,f\}$. If $\beta^*_k$ is applied in layer $k$, the remaining optimal parameters found at layer $k$, $\beta^*(i|k)$, $i=k+1,\ldots,k+N_k-1$, comprise a feasible solution to the MPC of Eq. \eqref{eq:LMPCAlgorithm} at $k+1$. The optimal value of the objective function of Eq. \eqref{eq:LMPCAlgorithm} at $k+1$ is upper bounded by the value that it would take under this feasible solution as follows: 
\begin{equation}
    \begin{aligned}
        J_{k+1}^* \leq \sum_{i=k+1}^{k+N_k-1}\braket{\tilde{\psi}_{i}|H_p|\tilde{\psi}_{i}}_{\beta^*(i|k)}\label{eq:optsln_kplus1}
    \end{aligned}
\end{equation}
Subtracting Eq. \eqref{eq:optsln_k} from Eq. \eqref{eq:optsln_kplus1}, the difference between the optimal values of the objective function in consecutive layers is bounded as follows:
\begin{equation}
    J_{k+1}^* - J_k^* \leq - \braket{\psi_{k}|H_p|\psi_{k}}_{\beta^*_k} \label{eq:5-100}
\end{equation}
for layers $k \in \{f-N+1,\dots,f-1\}$. Taking the sum of Eq. \eqref{eq:5-100} from $k=f-N+1$ when the shrinking horizon starts to $k=f-1$ and subtracting the optimal solution at $k=f$, i.e., $J_{f}^*$, we obtain:
\begin{equation}
    \begin{aligned}
        \sum_{k=f-N+1}^{f-1}& [J_{k+1}^*  - J_k^*] - J_{f}^*   \\ 
        & \leq \sum_{k=f-N+1}^{f-1} -\langle \psi_{k}|H_p|\psi_{k}\rangle_{\beta_k^*} - J_{f}^* \\
        & = \sum_{k=f-N+1}^{f} -\langle \psi_{k}|H_p|\psi_{k}\rangle_{\beta_k^*}\label{eq:DiffDuringShrinkingHorizon}
    \end{aligned}
\end{equation}
The sum of the differences between the optimal objective function values at layers $1$ through $f-1$, subtracting $J_{f}^*$, is given by:
\begin{equation}
    \begin{aligned}
        & \sum_{k=1}^{f-1} [J_{k+1}^* - J_k^*] - J_{e,f}^* = \sum_{k=1}^{f-N} [J_{k+1}^* - J_k^*] \\
        &\quad \quad + \sum_{k=f-N+1}^{f-1} [J_{k+1}^* - J_k^*] - J_{f}^* \\
        & \leq \sum_{k = N}^{f} \braket{\phi_{k}|H_p|\phi_{k}}_{\nu_k} - \sum_{k = 1}^{f-N} \braket{\psi_{k}|H_p|\psi_{k}}_{\beta^*_k} \\
        & \quad \quad - \sum_{k = f-N+1}^{f} \braket{\psi_{k}|H_p|\psi_{k}}_{\beta^*_k}\label{eq:FullSum}
    \end{aligned}
\end{equation}
where the upper bound results from substituting the right-hand side of Eq. \eqref{eq:DiffBeforeShrinkingHorizon} for $\sum_{k=1}^{f-N} [J_{k+1}^* - J_k^*]$ and the right-hand side of Eq. \eqref{eq:DiffDuringShrinkingHorizon} for $\sum_{k=f-N+1}^{f-1} [J_{k+1}^* - J_k^*] - J_{f}^*$. Furthermore, the sum of the differences between the optimal objective function values in consecutive layers from layer $1$ to layer $f-1$, subtracting the optimal value at layer $f$, can be written as follows:
\begin{equation}
    \begin{aligned}
        & \sum_{k=1}^{f-1} [J_{k+1}^* - J_k^*] - J_{f}^* \\
        & = J_{f}^* - J_{1}^* - J_{f}^* \\
        & = - J_{1}^* \label{eq:SumInequality}
    \end{aligned}
\end{equation}
The objective function under the optimal gate parameters $\beta^*(i|k)$ for $i\in \{k,\ldots,k+N-1\}$ for $k=1$ is upper bounded by the objective function value under the FALQON gate parameters, which is a feasible solution for Eq. \eqref{eq:LMPCAlgorithm:Objective} satisfying Eq. \eqref{eq:LMPCAlgorithm:TerminalConstraint}, as follows:
\begin{equation}
\begin{aligned}
    J_{1}^* \leq \sum_{k=1}^{N}\braket{\phi_{k}|H_p|\phi_{k}}_{\nu_k} \label{eq:Falqonbound}
\end{aligned}
\end{equation}
Taking the negative of both sides of Eq. \eqref{eq:Falqonbound}, combining the result with Eq. \eqref{eq:SumInequality}, and then using the bound in Eq. \eqref{eq:FullSum}, gives:
\begin{equation}
\begin{aligned}
      -\sum_{k=1}^{N}\braket{\phi_{k}|H_p|\phi_{k}}_{\nu_k} & \leq \sum_{k = N}^{f} \braket{\phi_{k}|H_p|\phi_{k}}_{\nu_k} \\
     & - \sum_{k = 1}^{f} \braket{\psi_{k}|H_p|\psi_{k}}_{\beta^*_k}
\end{aligned}
\end{equation}
Adding $\sum_{k=1}^{N}\braket{\phi_{k}|H_p|\phi_{k}}_{\nu_k}$ and $\sum_{k = 1}^{f} \braket{\psi_{k}|H_p|\psi_{k}}_{\beta^*_k}$ to both sides, the performance bound for the shrinking horizon case is
\begin{equation}
    \begin{aligned}
        \sum_{k=1}^{f}\braket{\psi_{k}|H_p|\psi_{k}}_{\beta^*_k} \leq \sum_{k=1}^{f}\braket{\phi_{k}|H_p|\phi_{k}}_{\nu_k}
    \end{aligned}
\end{equation}

\begin{rem}
    We highlight several points regarding Theorem~\ref{thm:terminalconstraintLMPC}.  First, we note that the summation on the right-hand side of Eq. \eqref{eq:prop1} is over a greater number of layers than the summation on the left-hand side. Furthermore, the case without a shrinking horizon used in the proof of Eq. \eqref{eq:prop1} does not necessarily enforce $\ket{\psi_{f}} = \ket{\phi_{f}}$ (only that $\ket{\tilde{\psi}_{f+N-1}} = \ket{\phi_{f+N-1}}$). Therefore, the final solution of the MPC at layer $k=f$ may not be as optimal as the solution of FALQON at the same layer when a fixed-length prediction horizon is used.  Employing a shrinking horizon, the expression in Eq. \eqref{eq:prop1Option2} can be obtained, where the summations on both sides of the inequality involve the same number of layers. Furthermore, with the shrinking horizon, $\ket{\psi_f} = \ket{\phi_f}$.
\end{rem}
\begin{rem}
This proof exploits an objective function structure consisting of a sum of functions of the same form that are lower-bounded and depend only on the quantum state and gate parameter at a given layer $k$.  They do not explicitly require that the form of this lower-bounded function must be $\tilde{\braket{H_p}}_{k}$ at layer $k$.  Therefore, the results of this proof will hold even with replacing $\langle \psi_k | H_p | \psi_k \rangle$ with a general lower-bounded function depending on $\ket{\psi_k}$ and $\beta_k$ at every layer, denoted by $l(\ket{\psi_k},\beta_k)$, meaning that the cumulative cost objective function would have the form $\sum_{i=k}^{k+N-1} l(\ket{\tilde{\psi}_i},\beta_i)$.
\end{rem}

\subsection{Pauli Propagation Example Including Truncations}\label{app:PauliPropExample}

In this section, we demonstrate how Pauli propagation and truncation methods can be used with the MPC-based algorithm through an explicit illustrative example. We consider the Max-Cut problem for the 4-node graph in Panel (a) of Figure~\ref{fig:MaxCutGraphs}, with a problem Hamiltonian of the following form:
\begin{equation*}
    H_p = -\frac{1}{2}(5I - Z_1Z_2 - Z_2Z_3 - Z_3Z_4 - Z_1Z_3 - Z_1Z_4)
\end{equation*}
and the following driver Hamiltonian:
\begin{equation*}
    H_d = X_1 + X_2 + X_3 + X_4
\end{equation*}

We first showcase the form of Eq. \eqref{eq:PauliPropEq1} applied to the computation of $\langle \tilde{H}_p \rangle_{k+N-1}$ for this example.  For simplicity of presentation, we consider $N=1$.  Then, $\tilde{\braket{H_p}}_{k+N-1} = \tilde{\braket{H_p}}_{k}$ is given by:
\begin{equation}
\begin{aligned}
    &\langle \tilde{\psi}_{k}| H_p | \tilde{\psi}_{k}\rangle = \\
    & \langle \psi_{k-1} | U_p^{\dagger} U_d^{\dagger}(\beta_k) H_p U_d(\beta_k)U_p| \psi_{k-1} \rangle \\
    & = -\frac{1}{2}\langle \psi_{k-1} | U_p^{\dagger} U_d^{\dagger}(\beta_k) 5I U_d(\beta_k)U_p| \psi_{k-1} \rangle \\
    & +\frac{1}{2}\langle \psi_{k-1} | U_p^{\dagger} U_d^{\dagger}(\beta_k) Z_1Z_2 U_d(\beta_k)U_p| \psi_{k-1} \rangle \\
    & +\frac{1}{2}\langle \psi_{k-1} | U_p^{\dagger} U_d^{\dagger}(\beta_k) Z_2Z_3 U_d(\beta_k)U_p| \psi_{k-1} \rangle \\
    & +\frac{1}{2}\langle \psi_{k-1} | U_p^{\dagger} U_d^{\dagger}(\beta_k) Z_3Z_4 U_d(\beta_k)U_p| \psi_{k-1} \rangle \\
    & +\frac{1}{2}\langle \psi_{k-1} | U_p^{\dagger} U_d^{\dagger}(\beta_k) Z_1Z_3 U_d(\beta_k)U_p| \psi_{k-1} \rangle \\
    & +\frac{1}{2}\langle \psi_{k-1} | U_p^{\dagger} U_d^{\dagger}(\beta_k) Z_1Z_4 U_d(\beta_k)U_p| \psi_{k-1} \rangle
    \end{aligned} \label{eq:EndCircuit}
\end{equation}
This is equivalent to a sum of expectation values of Pauli operators evaluated at the state predicted at the end of the circuit.  By applying Eq. \eqref{eq:PauliPropagation} to each term in Eq.~\eqref{eq:EndCircuit}, we can write each term as a linear combination of Pauli strings following Eq. \eqref{eq:Pauliexp}.  We will demonstrate the application of Eq. \eqref{eq:PauliPropagation} for one of these terms, the one involving $Z_2Z_3$, as an example.  Similar steps could then be performed for the others.

For each term, application of Eq. \eqref{eq:PauliPropagation} requires the explicit definition of the operators $U_p$ and $U_d$ for this example, which are given as follows:
\begin{equation}
    \begin{aligned}
        U_p &= e^{-iH_p\Delta_t} \\
        &= e^{-0.5iZ_1Z_2\Delta_t}e^{-0.5iZ_2Z_3\Delta_t}e^{-0.5iZ_3Z_4\Delta_t}\\
        &\quad \quad \quad e^{-0.5iZ_1Z_3\Delta_t}e^{-0.5iZ_1Z_4\Delta_t} \\
        & = U_{Z_1Z_2}(0.5\Delta_t)U_{Z_2Z_3}(0.5\Delta_t)U_{Z_3Z_4}(0.5\Delta_t) \\
        &\quad \quad \quad U_{Z_1Z_3}(0.5\Delta_t)U_{Z_1Z_4}(0.5\Delta_t) \label{eq:UpInExplicitExample}
    \end{aligned}
\end{equation}
\begin{equation}
    \begin{aligned}
        U_d &= e^{-iH_d \beta_k \Delta_t} \\
        & = e^{-iX_1 \beta_k \Delta_t}e^{-iX_2 \beta_k \Delta_t}e^{-iX_3 \beta_k \Delta_t}e^{-iX_4 \beta_k \Delta_t} \\
        &= U_{X_1}(\beta_k\Delta_t)U_{X_2}(\beta_k\Delta_t)U_{X_3}(\beta_k\Delta_t)U_{X_4}(\beta_k\Delta_t) \label{eq:UdInExplicitExample}
    \end{aligned}
\end{equation}

To compute $\langle \psi_{k-1} | U_p^{\dagger} U_d^{\dagger}(\beta_k) Z_1Z_3 U_d(\beta_k)U_p| \psi_{k-1} \rangle$ in Eq.~\eqref{eq:EndCircuit} following Eq. \eqref{eq:PauliPropagation}, we first conjugate $Z_2Z_3$ by the first operator in $U_d$ as follows:
\begin{equation}
    U^\dagger_{X_1}(\beta_k\Delta_t)Z_2Z_3 U_{X_1}(\beta_k\Delta_t) = Z_2Z_3\label{eq:PP1}
\end{equation}
since $[X_1,Z_2Z_3] = 0$. Conjugating $U^\dagger_{X_1}(\beta_k\Delta_t)Z_2Z_3 U_{X_1}(\beta_k\Delta_t)$ in Eq. \eqref{eq:PP1} by the next operator in $U_d$ in Eq. \eqref{eq:UdInExplicitExample}, which is $U_{X_2}(\beta_k\Delta_t)$, gives:
\begin{equation}
    \begin{aligned}
    U^\dagger_{X_2}&(\beta_k\Delta_t)Z_2Z_3 U_{X_2}(\beta_k\Delta_t) \\
    & = \cos(2\beta_k\Delta_t)Z_2Z_3 - i\sin(2\beta_k\Delta_t)Y_2Z_3 \label{eq:PP2}
    \end{aligned}
\end{equation}
demonstrating an example of observable \textit{branching}. Next, the expression in Eq. \eqref{eq:PP2} is conjugated by the next operator in $U_d$, which is $U_{X_3}$.  This produces a summation of the following two expressions:
\begin{equation}
    \begin{aligned}
    (\cos(2\beta_k\Delta_t))&U^\dagger_{X_3}(\beta_k\Delta_t)Z_2Z_3 U_{X_3}(\beta_k\Delta_t) \\
    & = \cos^2(2\beta_k\Delta_t)Z_2Z_3 \\
    &\quad- \cos(2\beta_k\Delta_t)\sin(2\beta_k\Delta_t)Z_2Y_3\label{eq:PP3}
    \end{aligned}
\end{equation}
\begin{equation}
    \begin{aligned}
    -(i\sin(2\beta_k\Delta_t))&U^\dagger_{X_3}(\beta_k\Delta_t)Y_2Z_3 U_{X_3}(\beta_k\Delta_t) \\
    &= -i\sin(2\beta_k\Delta_t)\cos(2\beta_k\Delta_t)Y_2Z_3 \\
    &\quad- \sin^2(2\beta_k\Delta_t)Y_2Y_3\label{eq:PP4}
    \end{aligned}
\end{equation}
These expressions demonstrate further operator branching.  Since $X_4$ commutes with all four of the Pauli operators from Eqs. \eqref{eq:PP3} and \eqref{eq:PP4}, i.e., $Z_2Z_3, Z_2Y_3, Y_2Z_3$, and $Y_2Y_3$, the summation of Eqs. \eqref{eq:PP3} and \eqref{eq:PP4} will be unaffected by conjugation by the final term in the expression for $U_d$ in Eq. \eqref{eq:UdInExplicitExample}, which is $U_{X_4}(\beta_k\Delta_t)$.  Thus, to develop the full summation representing $\langle \psi_{k-1} | U_p^{\dagger} U_d^{\dagger}(\beta_k) Z_1Z_3 U_d(\beta_k)U_p| \psi_{k-1} \rangle$ in the form of Eq. \eqref{eq:Pauliexp} by following Eq. \eqref{eq:PauliPropagation} would require next conjugating the sum of Eqs. \eqref{eq:PP3}-\eqref{eq:PP4} by each of the operators comprising $U_p$ in Eq. \eqref{eq:UpInExplicitExample}.  For the sake of brevity, we will not continue to demonstrate this for each of the four operators, but only with one term in Eq. \eqref{eq:PP4} for the first, which is $U_{Z_1Z_2}(0.5 \Delta_t)$, to showcase how this \textit{backpropagation} procedure can generate high-weight Pauli operators.  Conjugating $-i\sin(2\beta_k\Delta_t)\cos(2\beta_k\Delta_t)Y_2Z_3$ by $U_{Z_1Z_2}(0.5 \Delta_t)$ gives:  
\begin{equation}
    \begin{aligned}
    -i\sin(2\beta_k\Delta_t) &\cos(2\beta_k\Delta_t) U^\dagger_{Z_1Z_2}(0.5\Delta_t)Y_2Z_3 U_{Z_1Z_2}(0.5\Delta_t) \\
    & = -i\sin(2\beta_k\Delta_t) \cos(2\beta_k\Delta_t)\cos(\Delta_t)Y_2Z_3 \\
    & + \sin(2\beta_k\Delta_t) \cos(2\beta_k\Delta_t)\sin(\Delta_t)Z_1X_2Z_3\label{eq:PP5}
    \end{aligned}
\end{equation}
where the weight-3 observable $Z_1X_2Z_3$ is generated.  Thus, even with $N=1$, the process of writing the expression for $\langle \tilde{H}_p \rangle_k$ using Eq. \eqref{eq:PauliPropagation} may result in an expression for $\langle \tilde{H}_p \rangle_k$ containing Pauli operators of higher weight than those in $H_p$.  Particularly as $N$ is increased, increasing the number of times that operator conjugation is performed in the backpropagation procedure, observables with weight up to $n$ may be generated, where $n$ is the number of qubits.  For example, if $N$ had been greater than 1 in this example, since a term involving $Z_1X_2Z_3$ will persist through the remaining conjugations involving operators from $U_p$, $Z_1X_2Z_3$ would be conjugated by $U_{X_3}(\beta\Delta_t)$ in layer 2, causing the observable $Z_1X_2Y_3$ to be generated, which after evolution through $U_{Z_3Z_4}(0.5\Delta_t)$ would cause the weight-4 observable $Z_1X_2X_3Z_4$ to be generated. This effect can cause up to $4^n-1$ observables to be generated in the worst case if the circuit is sufficiently long.  This offers opportunity for truncation strategies to reduce the number of terms in the expression for $\langle \tilde{H}_p \rangle_{k+N-1}$.

The structure of the backpropagated form of Eq.~\eqref{eq:EndCircuit} is dependent on $N$, but the gate parameters affect the coefficients, as can be seen in Eqs.~\eqref{eq:PP1}-\eqref{eq:PP5} above.  For this reason, using weight-based truncations results in a fixed structure for $\tilde{\braket{H_p}}_{k+N-1}$ for a given $N$, regardless of the values of the gate parameters, whereas with coefficient-based truncations, changing the gate parameters can affect which terms in the summation obtained from backpropagation are trimmed. Truncations are applied after each application of Eq. \eqref{eq:PauliPropagation}, e.g., a weight-based truncation where Pauli operators greater than weight-2 are trimmed would be applied to Eq. \eqref{eq:PP5} to trim the $Z_1X_2Z_3$ term, resulting in $-i\sin(2\beta_k\Delta_t) \cos(2\beta_k\Delta_t) U^\dagger_{Z_1Z_2}(0.5\Delta_t)Y_2Z_3 U_{Z_1Z_2}(0.5\Delta_t)$ $= -i\sin(2\beta_k\Delta_t) \cos(2\beta_k\Delta_t)\cos(\Delta_t)Y_2Z_3$. With respect to coefficient-based truncations, if $\Delta_t = 0.001$ and $\beta_k = -0.1$, for example, in Eq. \eqref{eq:PP5}, $Y_2Z_3$ would have a coefficient of approximately $-2\times 10^{-4}i$, and $Z_1X_2Z_3$ would have a coefficient of approximately $-2\times 10^{-7}$. If a threshold for coefficient-based truncation is chosen such that the coefficient of $Z_1X_2Z_3$, which has the smaller magnitude, is below the threshold but that for $Y_2Z_3$ is above it, then $Z_1X_2Z_3$ would be trimmed, again resulting in $-i\sin(2\beta_k\Delta_t) \cos(2\beta_k\Delta_t) U^\dagger_{Z_1Z_2}(0.5\Delta_t)Y_2Z_3 U_{Z_1Z_2}(0.5\Delta_t)$ $= -i\sin(2\beta_k\Delta_t) \cos(2\beta_k\Delta_t)\cos(\Delta_t)Y_2Z_3$. 

\begin{figure*}
    \centering
    \begin{subfigure}{0.3\textwidth}
        \centering
        \begin{overpic}[width=0.9\textwidth]{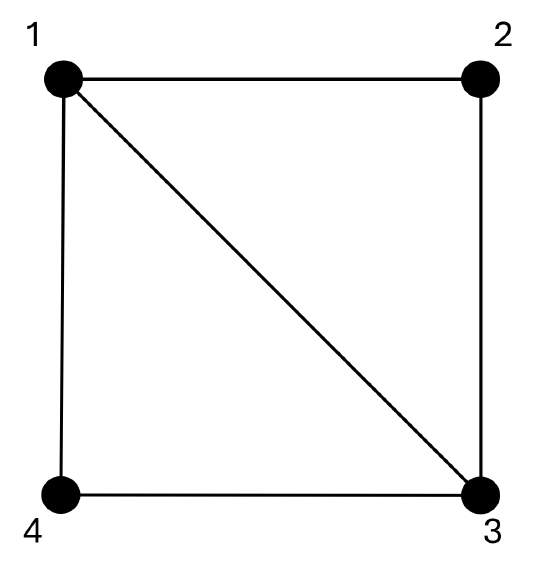} 
            \put(-7,90){\textbf{(a)}}
        \end{overpic}
    \end{subfigure}
    \begin{subfigure}{0.3\textwidth}
        \centering
        \begin{overpic}[width=1.0\textwidth]{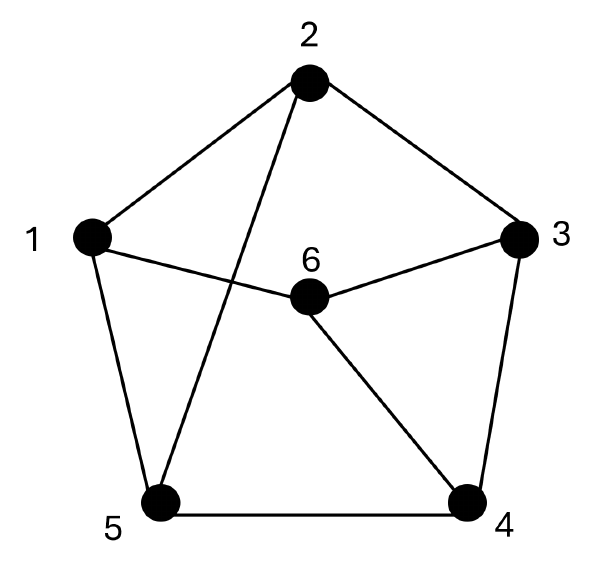} 
            \put(5,87){\textbf{(b)}}
        \end{overpic}
    \end{subfigure}
    \begin{subfigure}{0.3\textwidth}
        \centering
        \begin{overpic}[width=1.0\textwidth]{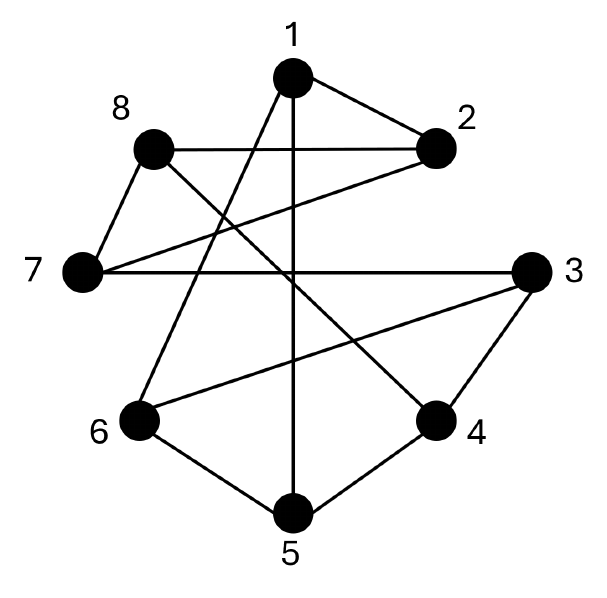} 
            \put(5,87){\textbf{(c)}}
        \end{overpic}
    \end{subfigure}
    \caption{\textbf{Graphs for Max-Cut problems.} Each of the graphs for the Max-Cut problems used in the simulations in this work are shown with vertex indices labeled. Panel (a) is the 4-node graph used in Figures~\ref{fig:Weighted4NodeMaxcut} and~\ref{fig:4NodeVaryingHorizons}, Panel (b) is the 6-node graph used in Figure~\ref{fig:Heatmaps}, and Panel (c) is the 8-node graph used in Figures~\ref{fig:8NodeSims} and~\ref{fig:FluctuationAnalysis}.}
    \label{fig:MaxCutGraphs}
\end{figure*}

\subsection{Max-Cut Simulations} \label{sec:MaxCutWeights}
Figure~\ref{fig:MaxCutGraphs} presents the graphs used in the Max-Cut simulations in this work. In these examples, each weight was generated randomly from a uniform distribution between 0 and 2. The weights used in each run from Figure~\ref{fig:Weighted4NodeMaxcut} are shown in Table~\ref{tab:4NodeWeights}. The problem Hamiltonian for the weighted 6-node Max-Cut problem used in Figure~\ref{fig:Heatmaps} is defined as follows:
 \begin{equation}
 \begin{aligned}
     H_p & = 1.30188113Z_1Z_2 + 1.45309629Z_1Z_5 \\
     & + 1.33858841Z_1Z_6 + 1.88765732Z_2Z_3 \\
     & + 1.51154919Z_2Z_5 + 0.58080955Z_3Z_4 \\
     & + 1.16651376Z_3Z_6 + 0.32684262Z_4Z_5 \\
     & + 0.0093935Z_4Z_6
\end{aligned}
\end{equation}

\begin{table}[]
     \centering
     \resizebox{\columnwidth}{!}
     {\begin{tabular}{c|c|c|c|c|c}
          Run & $Z_1Z_2$ & $Z_2Z_3$ & $Z_3Z_4$ & $Z_1Z_4$ & $Z_1Z_3$ \\
          1 & 1.798525 & 0.920494 & 1.142551 & 0.832072 & 1.981212 \\
          2 & 1.598644 & 0.112369 & 1.152085 & 0.035653 & 1.832361 \\
          3 & 1.037419 & 0.783750 & 0.722643 & 1.151326 & 1.093111 \\
          4 & 1.943770 & 1.581998 & 1.379596 & 0.804860 & 1.355905 \\
          5 & 1.803008 & 1.397768 & 0.658868 & 0.916864 & 0.690076 \\
          6 & 1.364343 & 0.374132 & 1.067018 & 1.611811 & 0.044410 \\
          7 & 0.141521 & 1.254389 & 0.491707 & 1.646971 & 1.080704 \\
          8 & 1.548532 & 0.295641 & 1.060026 & 1.779830 & 0.592130 \\
          9 & 0.969662 & 0.272750 & 0.106999 & 1.185751 & 0.024157 \\
          10 & 1.930573 & 1.074242 & 0.113385 & 0.083101 & 0.016078
     \end{tabular}}
     \caption{\textbf{4-node Max-Cut weights:} The weights for each of the ten instances from Figure~\ref{fig:Weighted4NodeMaxcut}, generated randomly from a uniform distribution between 0 and 2, where the first six digits after the decimal are noted.  The weight is presented in the column corresponding to the Pauli operator with which it is associated.}
     \label{tab:4NodeWeights}
\end{table}

\begin{figure}
    \centering
    \includegraphics[width=1.08\columnwidth]{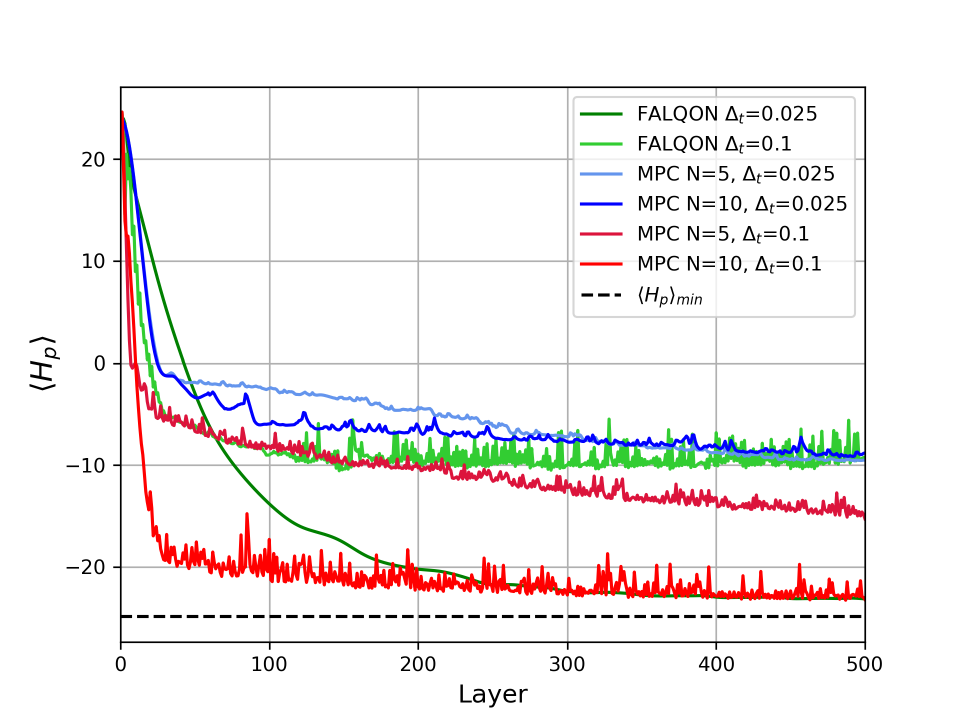}
    \caption{\textbf{Transverse field Ising model simulations.} We simulate the MPC-based quantum algorithm with various values of $\Delta_t$ and $N$ and FALQON with $\Delta_t = 0.025$ and $\Delta_t = 0.1$ on a 4x2 transverse field Ising model. The FALQON trajectories are shown in green, while the $\langle H_p \rangle$ trajectories under the MPC-based algorithm with $\Delta_t = 0.025$ are shown in blue, and the $\langle H_p \rangle$ trajectories under the MPC-based algorithm with $\Delta_t = 0.1$ are shown in red.  In the trajectory of $\langle H_p \rangle$ using the MPC-based algorithm with $\Delta_t = 0.1$ and a prediction horizon of $N=10$, $\braket{H_p}$ decreases quickly over the first 50 layers, then continues to slowly decrease while fluctuating, where at some layers, $\braket{H_p}$ reaches lower values than the smallest value that FALQON with $\Delta_t = 0.025$ achieves over 500 layers.}
    \label{fig:TFIMexact}
\end{figure}

\subsection{Analysis of Fluctuating Trajectories}\label{sec:FluctuationAnalysis}
In this section, we provide additional analyses regarding the fluctuations in $\braket{H_p}$ observed in many of the figures demonstrating the use of the MPC-based algorithm, such as between layers 100 and 150 for $N=2$ in Figure~\ref{fig:4NodeVaryingHorizons}.  The existence of fluctuations in $\braket{H_p}$ between layers is a feature of the terminal cost MPC that persists across various $N$ and $\Delta_t$ values, as well as various truncations when Pauli propagation is used. These fluctuations are demonstrated for a range of $N$ and $\Delta_t$ for the TFIM shown in Figure~\ref{fig:TFIMexact} where full state vector simulation is used in Eq.~\eqref{eq:MPCAlgorithm:Model}, and furthermore shown in Figure \ref{fig:TFIMPauliProp} when Pauli propagation is used for the same problem. 

The terminal cost MPC trajectories in Figure~\ref{fig:TFIMexact} can be contrasted with the performance of FALQON, which prescribes a smooth trajectory with $\Delta_t = 0.025$, but a non-monotonically decreasing and fluctuating trajectory when using $\Delta_t = 0.1$, indicating that $\Delta_t$ is too large. The fluctuations seen in the four instances of the MPC-based algorithm are different in character than those observed with FALQON using $\Delta_t = 0.1$, however, as the MPC-based algorithm trajectories show an overall decrease in $\langle H_p \rangle$ across the 500 layers simulated, suggesting that the fluctuations contribute to this decrease. In the following sections we provide evidence that the fluctuations are a desirable feature of the terminal cost MPC.

\paragraph{Analysis on a Max-Cut Problem}
Panel (a) in Figure~\ref{fig:FluctuationAnalysis} provides an analysis of the fluctuations when the terminal cost MPC is applied to the Max-Cut problem over the 8-node graph in Panel (c) of Figure~\ref{fig:MaxCutGraphs}. To investigate the role of the prediction horizon in the behavior of the fluctuations, we simulate an instance of the terminal cost MPC as a baseline, and initialize FALQON and four additional instances of the terminal cost MPC from a state of the baseline MPC prior to a period of fluctuations, $\ket{\psi_{225}}$. At layer 227, the trajectory of the baseline MPC starts to increase before a period of fluctuations lasting until layer 281, after which the trajectory of $\langle H_p \rangle$ settles into a state where it is lower than before the fluctuations began. The terminal cost MPC with $N=1$ and FALQON produce trajectories of $\langle H_p \rangle$ that overlay one another, where after an initial small decrease in $\langle H_p \rangle$ over the first two layers, both trajectories remain flat. In contrast, the terminal cost MPCs with $N=2,3,5,$ and $10$ all increase the value of $\langle H_p \rangle$ to then subsequently decrease it below what is achieved with the MPC-based algorithm with $N=1$ and with FALQON. As $N$ increases, the fluctuations grow larger, their frequency decreases, and $\langle H_p \rangle$ reaches a lower value with increasing $N$ between layers 280 and 290.  These results provide evidence that the fluctuations are being used productively by the MPC-based algorithm.

\paragraph{Analysis on a Transverse Field Ising Model}
Here we consider again the FALQON simulation for the 4x2 TFIM when $\Delta_t = 0.025$, this time over 1000 layers shown in Panel (b) of Figure~\ref{fig:FluctuationAnalysis}, noting that $\braket{H_p}$ only decreases on the order of $10^{-2}$ from layer 500 to layer 1000. To investigate whether the terminal cost MPC would improve on FALQON's solution, and whether fluctuations would be involved if so, we initialize the terminal cost MPC at layer 950 and simulate for 50 layers. The terminal cost MPC first increases the value of $\langle H_p \rangle$ compared to the value of $\langle H_p \rangle$ at layer 950, but subsequently decreases it, though with continued fluctuations in $\langle H_p \rangle$. The gradual lowering of the minimum value of $\langle H_p \rangle$ observed over layers 950 to 1000 under the terminal cost MPC suggests that the fluctuations are being used productively, especially since it used a fluctuation to kick itself off of the FALQON solution toward improved values of $\langle H_p \rangle$ by the end of the circuit. This example also provides evidence that the MPC-based algorithm can be used in conjunction with FALQON in cases where the FALQON trajectory converges to a suboptimal solution, to increase solution quality.

\begin{figure*}[t]
    \centering
    \begin{subfigure}{0.49\textwidth}
        \centering
        \begin{overpic}[width=1.05\textwidth]{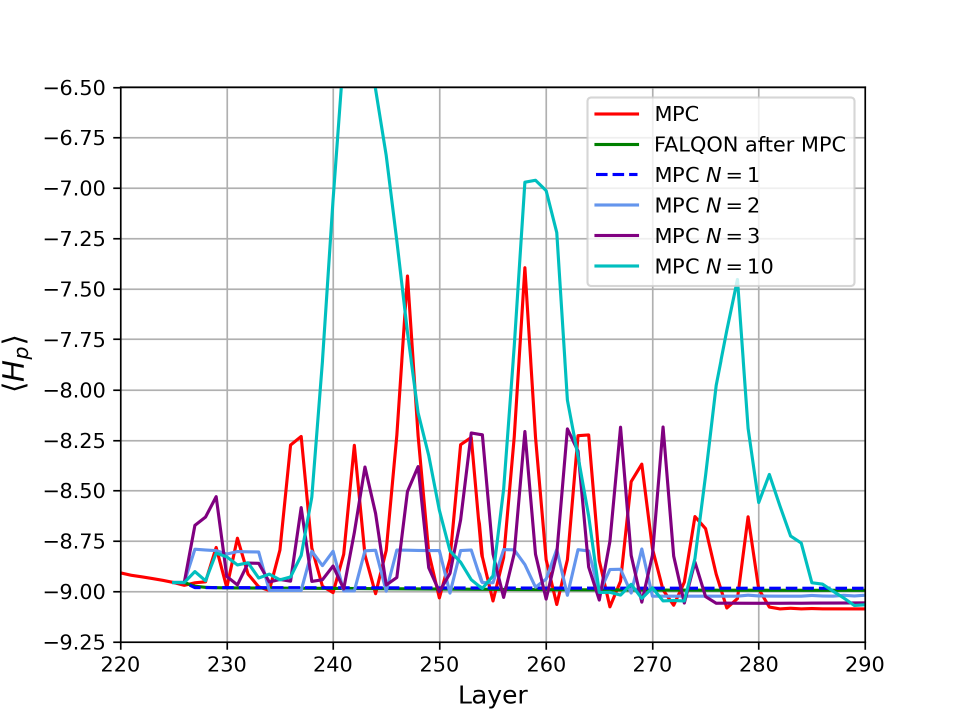}
            \put(0,70){\textbf{(a)}}
        \end{overpic}
    \end{subfigure}
    \begin{subfigure}{0.49\textwidth}
        \begin{overpic}[width=1.05\textwidth]{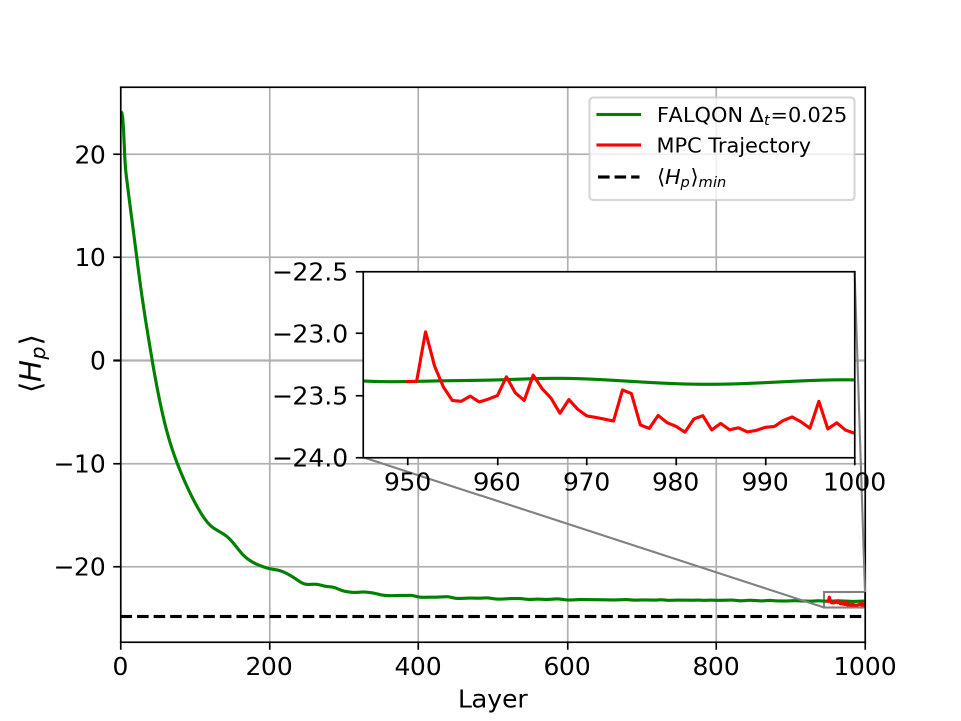}
            \put(0,70){\textbf{(b)}}
        \end{overpic}
    \end{subfigure}
    \caption{\textbf{Analysis of fluctuations.} In Panel (a), we simulate the MPC-based algorithm on the Max-Cut problem over an 8-node graph for 500 layers with $N=5$ and $\Delta_t = 0.02$, zooming in on a region where the trajectory fluctuates. We initialize FALQON and multiple instances of the MPC algorithm with varying prediction horizons, $N = 1,2,3,$ and 10, each with the step size $\Delta_t = 0.02$, with $\ket{\psi_{225}}$, and simulate each for 65 layers. In each instance of the terminal cost MPC, the initial parameter, $\beta^*_{226}$, is found via optimization. The step size $\Delta_t = 0.02$ was selected due to being a suitable step size when used in FALQON. Panel (b) displays a simulation of FALQON on the 4x2 TFIM over 1000 layers where $\Delta_t = 0.025$. The terminal constraint MPC is initialized at $\ket{\psi_{950}}$ from FALQON, where $N=5$, $\Delta_t = 0.025$, and $\beta^*_{951} = -A_{950}$.}
    \label{fig:FluctuationAnalysis}
\end{figure*}

\paragraph{Analysis of Predicted Trajectories}
While these studies illustrate that the fluctuations are indeed productive and calculated, the manner by which the terminal cost MPC determines how to produce productive fluctuations remains to be explained.  To this end, Figure~\ref{fig:OpenLoopTrajectories} shows a series of predictions that the terminal cost MPC is making as it produces the fluctuating trajectories shown in Panel (b) of Figure~\ref{fig:FluctuationAnalysis}. The actual trajectory of $\langle H_p \rangle$ under the terminal cost MPC is shown in red, and the predictions under the optimal gate parameters computed at each layer are shown as dashed-blue lines. Initially, $\beta^*_{951} = -A_{950}$. The terminal cost MPC then receives measurements of $\ket{\psi_{951}}$ and solves for $\beta^*_{952}\cdots \beta^*_{956}$, where the predicted trajectory of $\braket{H_p}$ under these values is shown in Panel (a) in Figure~\ref{fig:OpenLoopTrajectories}.

The terminal cost MPC predicts that by initially increasing the value of $\langle H_p \rangle$, a value of $\tilde{\braket{H_p}}_{956}$ that is more optimal than $\braket{H_p}_{951}$ can be found. According to the receding horizon policy of the MPC, $\beta_{952}^*$ is applied to increase $\langle H_p \rangle$, and the optimization problem is re-solved at the next layer. In Panel (b), the MPC initialized by $\ket{\psi_{952}}$, in seeking to minimize $\tilde{\braket{H_p}}_{957}$, sees that by selecting a different set of parameters than it had planned at layer 952 for layers 953 through 956, it will be able to obtain a value of $\langle H_p \rangle$ at layer 957 that is (locally) minimized.  This process of moving the horizon that initializes the optimization problem from a different state and affects the layer where $\langle H_p \rangle$ is minimized continues over the remainder of the plots shown in Figure~\ref{fig:OpenLoopTrajectories}. As the layers progress, this eventually results in decreases in the minimum value of $\langle H_p \rangle$ achieved within the circuit.  Since the MPC can only predict over a few layers at a time, the gate parameters that it plans often end up being significantly revised at future layers, such that the trajectory of $\langle H_p \rangle$ ends up in many cases to be quite different from the predicted trajectory from any one layer. Based on the analyses of this section, we see that the fluctuations observed in the trajectories of $\langle H_p \rangle$ in many of the simulations of the terminal cost MPC are caused by a combination of the terminal cost objective function, the receding horizon, and the value of $N$.

\begin{figure*}
    \centering
    \begin{subfigure}{0.32\textwidth}
        \centering
        \begin{overpic}[width=\textwidth]{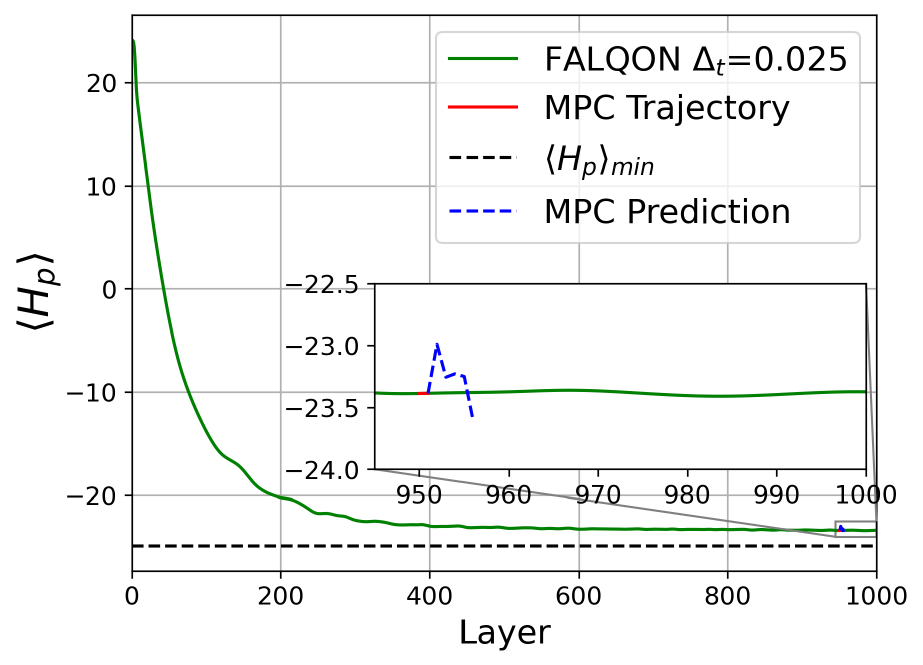}
            \put(0,65){\textbf{(a)}}
        \end{overpic}
    \end{subfigure}
    \begin{subfigure}{0.32\textwidth}
        \centering
        \begin{overpic}[width=\textwidth]{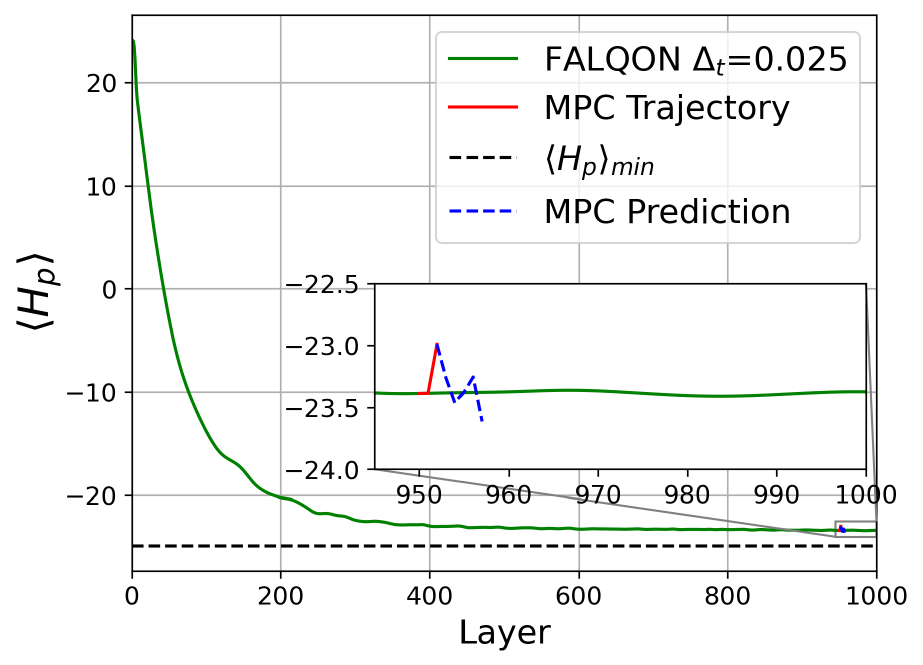}
            \put(0,65){\textbf{(b)}}
        \end{overpic}
    \end{subfigure}
    \begin{subfigure}{0.32\textwidth}
        \centering
        \begin{overpic}[width=\textwidth]{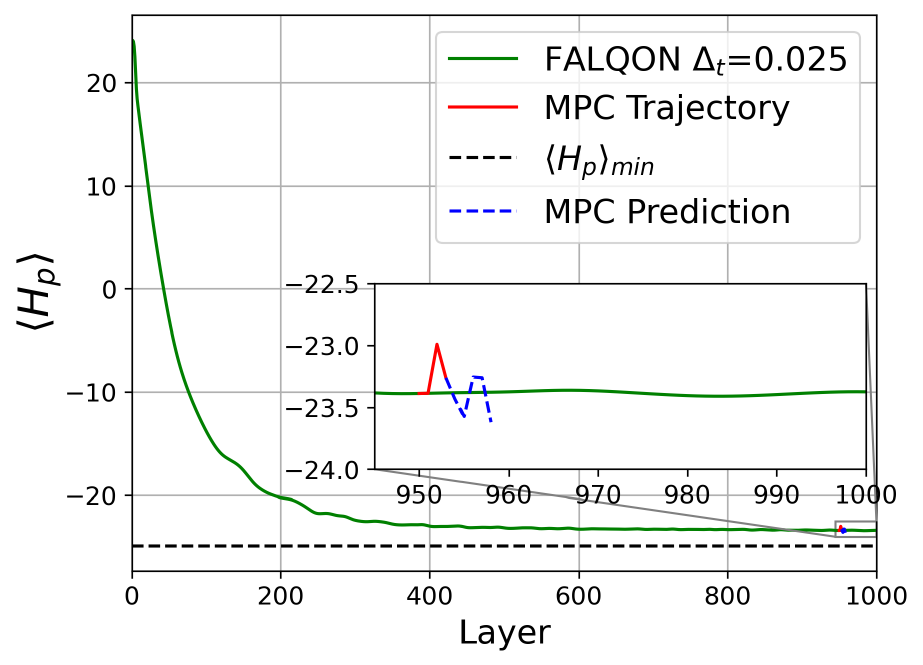}
            \put(0,65){\textbf{(c)}}
        \end{overpic}
    \end{subfigure}

    \bigskip
    \begin{subfigure}{0.32\textwidth}
        \centering
        \begin{overpic}[width=\textwidth]{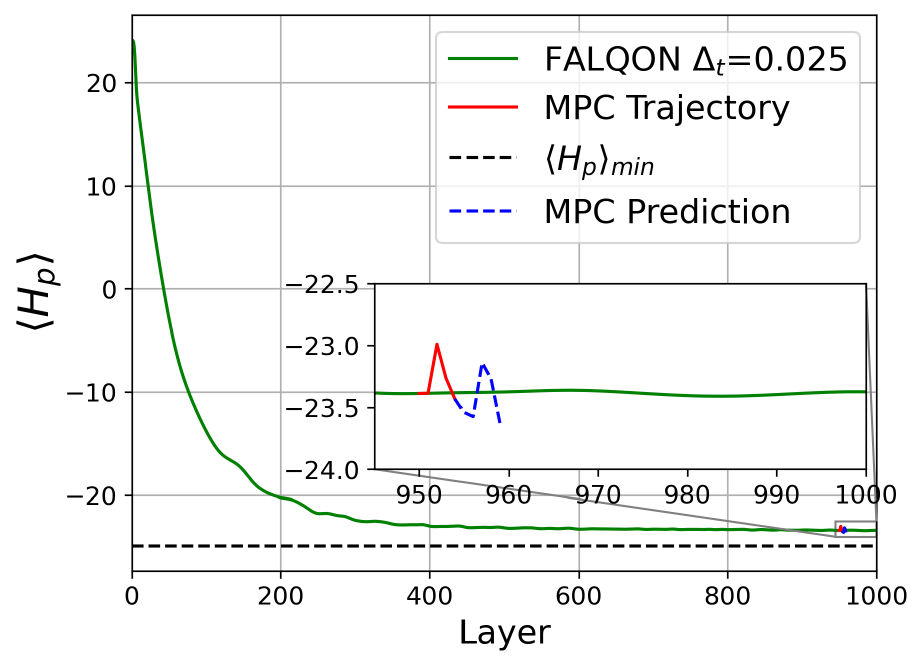}
            \put(0,65){\textbf{(d)}}
        \end{overpic}
    \end{subfigure}
    \begin{subfigure}{0.32\textwidth}
        \centering
        \begin{overpic}[width=\textwidth]{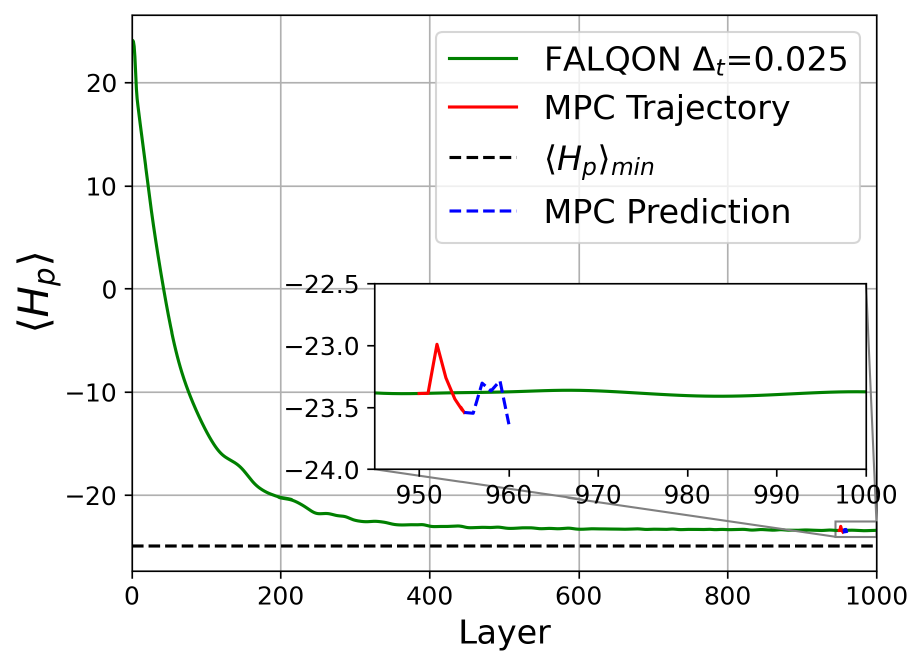}
            \put(0,65){\textbf{(e)}}
        \end{overpic}
    \end{subfigure}
    \begin{subfigure}{0.32\textwidth}
        \centering
        \begin{overpic}[width=\textwidth]{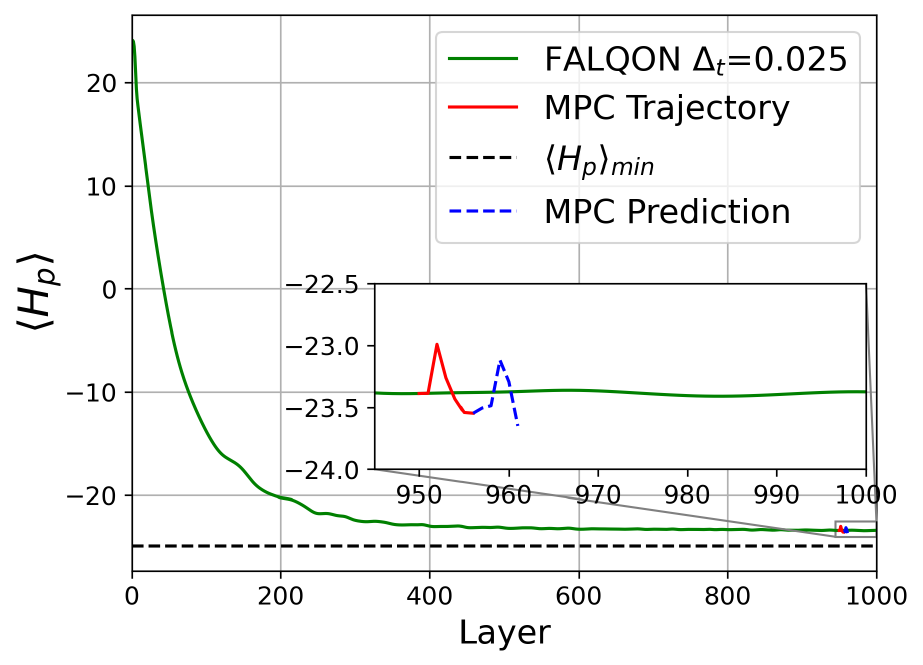}
            \put(0,65){\textbf{(f)}}
        \end{overpic}
    \end{subfigure}

    \bigskip
    \begin{subfigure}{0.32\textwidth}
        \centering
        \begin{overpic}[width=\textwidth]{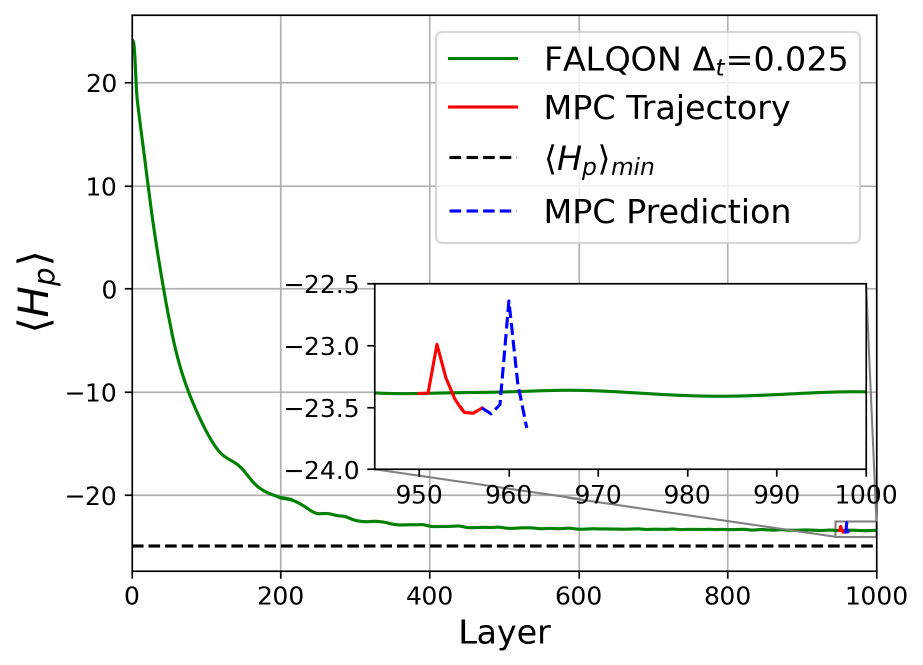}
            \put(0,65){\textbf{(g)}}
        \end{overpic}
    \end{subfigure}
    \begin{subfigure}{0.32\textwidth}
        \centering
        \begin{overpic}[width=\textwidth]{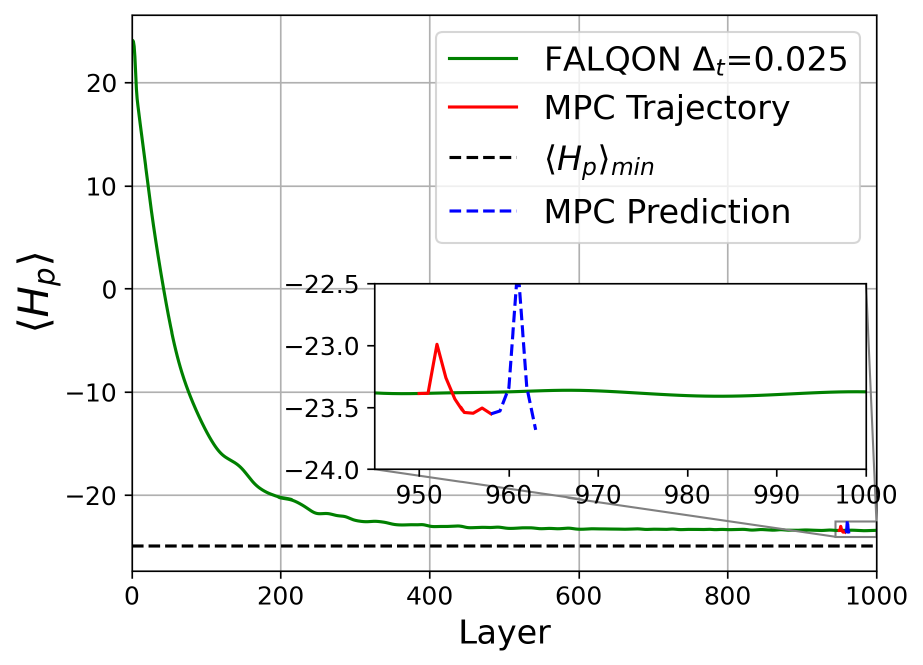}
            \put(0,65){\textbf{(h)}}
        \end{overpic}
    \end{subfigure}
    \begin{subfigure}{0.32\textwidth}
        \centering
        \begin{overpic}[width=\textwidth]{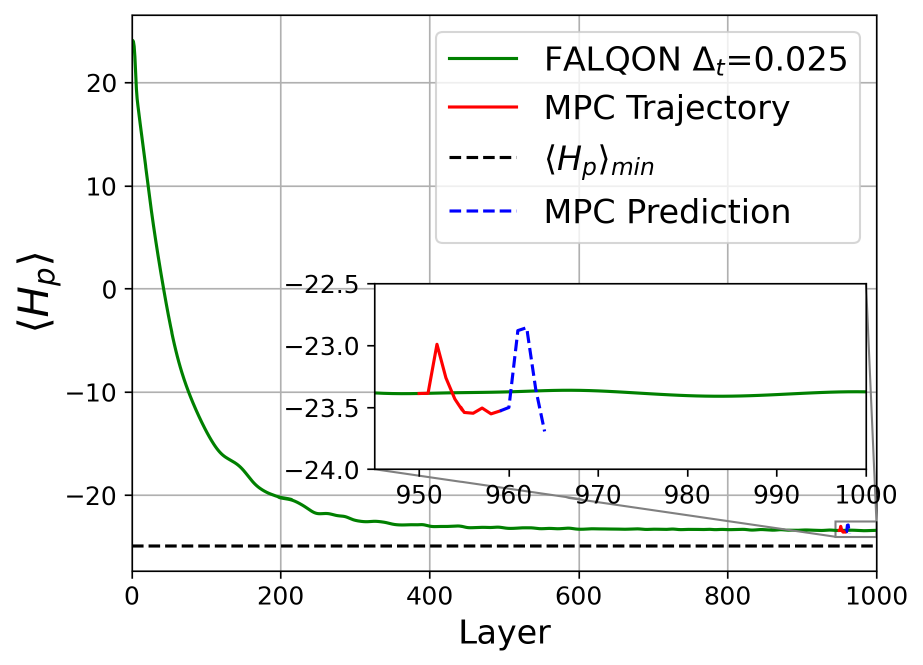}
            \put(0,65){\textbf{(i)}}
        \end{overpic}
    \end{subfigure}
    \caption{\textbf{Open-loop MPC trajectories.} Panels (a)-(i) display predicted trajectories determined at layers 952 through 960 of the terminal cost MPC shown in Panel (a) of Figure~\ref{fig:FluctuationAnalysis}, demonstrating the process behind the trajectory. In each of the figures, the actual trajectory of $\langle H_p \rangle$ is shown in red, whereas the predicted trajectory over the next five layers under the optimal gate parameters computed at the given layer is plotted as a dashed blue line. The minimum value of $\langle H_p \rangle$ for the 4x2 transverse field Ising model problem is plotted as a dashed black line and the trajectory under FALQON is provided for comparison.}
    \label{fig:OpenLoopTrajectories}
\end{figure*}

\end{document}